\documentclass[journal,final,twocolumn,10pt,twoside]{IEEEtranTCOM}
\normalsize

\usepackage{mathtools}
\usepackage{amsmath}
\usepackage{amssymb}
\usepackage{amsmath}
\usepackage{cite}
\usepackage{algorithm}
\usepackage{enumerate,multirow}
\usepackage{algpseudocode}

\usepackage{siunitx}
\usepackage{multirow}
\usepackage{multicol}
\usepackage{color}
\usepackage{xcolor}

\usepackage{comment}

\usepackage{graphicx}
\usepackage{graphicx}
\date{}

\usepackage{calc}
\newcolumntype{M}[1]{>{\centering\arraybackslash}m{#1}}
\newcolumntype{N}{@{}m{0pt}@{}}

\usepackage{mdwmath}
\usepackage{blindtext}
\usepackage{eqparbox}
\usepackage{fixltx2e}
\usepackage{stfloats}

\usepackage{amsthm}

\newtheorem{theorem}{{Theorem}}

\newtheorem{proposition}[theorem]{{Proposition}}

\newtheorem{definition}{{Definition}}

\DeclareMathAlphabet{\mathbfsl}{OT1}{ppl}{b}{it} 

\def\QEDclosed{\mbox{\rule[0pt]{1.3ex}{1.3ex}}} 

\def\QED{\QEDclosed} 

\def\endproof{\hspace*{\fill}~\QED\par\endtrivlist\unskip}

\newcommand{\be}[1]{\begin{equation}\label{#1}}
	\newcommand{\ee}{\end{equation}}

\renewcommand{\le}{\leqslant} 
\renewcommand{\leq}{\leqslant}

\newcommand{\Pref}[1]{Pro\-po\-si\-tion\,\ref{#1}}

\newcommand{\Cref}[1]{Co\-ro\-lla\-ry\,\ref{#1}}

\begin{document}
	\title{Machine Learning-Aided Efficient  Decoding of Reed-Muller Subcodes}
	\author{Mohammad Vahid Jamali, Xiyang Liu, Ashok Vardhan Makkuva, \\Hessam Mahdavifar, Sewoong Oh, and Pramod Viswanath
		\thanks{This paper was presented in part at the IEEE International Symposium on Information Theory (ISIT), Melbourne, Victoria, Australia, July 2021 \cite{jamali2021Reed}.}
		\thanks{Mohammad Vahid Jamali and Hessam Mahdavifar are with the Electrical Engineering and Computer Science Department at the University of Michigan, Ann Arbor, MI 48109, USA (e-mail: mvjamali@umich.edu, hessam@umich.edu). Xiyang Liu and Sewoong Oh are with the Paul G. Allen School of Computer Science \& Engineering at the University of Washington, Seattle, WA 98195, SUA (e-mail: xiyangl@cs.washington.edu, sewoong@cs.washington.edu). Ashok Vardhan Makkuva is with the School of Computer and Communication Sciences at EPFL, Lausanne, Switzerland (email: ashok.makkuva@epfl.ch). And 
		Pramod Viswanath is with the Department of Electrical and Computer Engineering at Princeton University,  Princeton, NJ 08540, USA (email: pramodv@princeton.edu).
		}
	}
\maketitle

\begin{abstract}
Reed-Muller (RM) codes achieve the capacity of general binary-input memoryless symmetric channels and { are conjectured to have a comparable performance to that of random codes in terms of scaling laws. However, such results are established assuming maximum-likelihood decoders for general code parameters. Also, RM codes only admit limited sets of rates. Efficient decoders such as successive cancellation list (SCL) decoder and recently-introduced recursive projection-aggregation (RPA) decoders are available for RM codes at finite lengths. In this paper, we focus on subcodes of RM codes with flexible rates. We first extend the RPA decoding algorithm to RM subcodes.}
To lower the complexity of our decoding algorithm, referred to as subRPA, we investigate different approaches to prune the projections.
Next, we derive the soft-decision based version of our algorithm, called soft-subRPA, that not only improves upon the performance of subRPA but also enables a differentiable decoding algorithm. Building upon the soft-subRPA algorithm, we then provide a framework for training a machine learning (ML) model to search for \textit{good} sets of projections that minimize the decoding error rate.
Training our ML model enables achieving very close to the performance of full-projection decoding with a significantly smaller number of projections. We also show that the choice of the projections in decoding  RM subcodes matters significantly, and our ML-aided projection pruning scheme is able to find a \textit{good} selection, i.e., with negligible performance degradation compared to the full-projection case, 
given a reasonable number of projections.
\end{abstract}
\section{Introduction}\label{Sec1}
\IEEEPARstart{R}EED-MULLER (RM) codes are among the first families of error-correcting codes, invented almost seven decades ago \cite{reed1954class,muller1954application}. They have received significant renewed interest after the breakthrough invention of polar codes \cite{arikan2009channel}, given the close connection between the two classes of codes. The generator matrices for both RM and polar codes can be obtained from the same square matrices -- the Kronecker powers of a $2 \times 2$ matrix -- though by different rules for selecting the rows. In fact, such a selection of rows for polar codes is channel-dependent but the RM encoder picks the rows with the largest Hamming weights, resulting in a universal construction.
RM codes are also conjectured to have characteristics similar to those of random codes in terms of both weight enumeration \cite{kaufman2012weight} and scaling laws \cite{hassani2018almost}. Moreover, 
 Reeves and Pfister have recently shown that RM codes achieve the capacity of general binary-input memoryless symmetric (BMS) channels \cite{reeves2021reed} under the bit maximum-a-posteriori (bit-MAP) decoding. This solves a long-standing open problem in coding theory while leaving the problem of finding efficient decoders for RM codes to provably achieve (or perform close to) such an excellent performance open. 

 Among the earlier results on decoding RM codes \cite{reed1954class,dumer2004recursive,dumer2006soft2,dumer2006soft,sakkour2005decoding,saptharishi2017efficiently,santi2018decoding}, Dumer's recursive list decoding algorithm \cite{dumer2004recursive,dumer2006soft2,dumer2006soft} provides a trade-off between the decoding complexity and the error probability. In other words, it is capable of achieving close to the maximum likelihood decoding performance for large enough, e.g., exponential in blocklength, list sizes. Recently, Ye and Abbe \cite{ye2020recursive} proposed a recursive projection-aggregation (RPA) algorithm for decoding RM codes. The RPA algorithm first projects the received corrupted codeword onto its cosets. It then recursively decodes the projected codes to, finally, construct the decoded codeword by properly aggregating the intermediate decoding results. 
Building upon the projection pruning idea in \cite{ye2020recursive}, a method for reducing the complexity of the RPA algorithm has also been explored in \cite{fathollahi2020sparse}. Moreover, a framework for encoding and decoding RM codes based on the product of smaller RM code components has been explored in \cite{jamali2021lowcomplexity}, with potential applications to low-capacity channels \cite{fereydounian2019channel}. Furthermore, building upon the computational tree of RM (and polar) codes, a class of neural encoders and decoders has been proposed in \cite{makkuva2021ko} via deep learning methods.

Besides lacking an efficient decoder in general, the structure of RM codes does not allow choosing a flexible rate. To clarify this, let $k$ and $n$ denote the code dimension and blocklength, respectively. Due to the underlying Kronecker product structure of RM codes, the code blocklength is a power of two, i.e., $n=2^m$, where $m$ is a design parameter. Additionally, RM codes posses another parameter $r$, that stands for the \textit{order} of the code, where $0\leq r\leq m$. Given the code blocklength $n$, one can then only construct RM codes with $m+1$ possible values for the code rate, each corresponding to a given code order $r$.

This research is inspired by the aforementioned two critical issues of RM codes. More specifically, we target subcodes of RM codes (with flexible rates that can take any code dimension from $1$ to $n$),
and our primary goal is to design low-complexity decoders for the RM subcodes. To this end, we first extend the RPA algorithm to what we call ``subRPA'' in this paper. Similar to the RPA algorithm, subRPA starts by projecting the received corrupted codeword onto the cosets. However, since the projected codes are no longer RM codes of lower orders, their corresponding generator matrices have different ranks (i.e., different code dimensions). SubRPA applies the MAP decoder at the bottom layer, which is feasible and efficient given the low dimension of the projected codes at that layer. It then aggregates the results back to recursively decode the received codeword. 

A major focus of this work is on reducing the complexity of our proposed decoding algorithms by pruning many of redundant projections. 
Through exploring different projection pruning strategies, we empirically show that the choice of projections can significantly impact the decoding performance of RM subcodes. 
We first propose a method, referred to as the \textit{minRank} projection pruning scheme (incurring the lowest decoding complexity, given a number of projections), that is observed to deliver a very good performance in a variety of scenarios. However, our results show that there are cases where even a random pruning scheme may outperform the minRank selection, especially when the number of projections used for the decoding are significantly smaller than the full number of projections.
Motivated by these observations, we leverage the recent advances in channel coding via machine/deep learning \cite{o2017introduction,gruber2017deep,jiang2019turbo,kim2018deepcode,kim2020physical,makkuva2021ko,akyildiz2022ml,jamali2021productae} to pick the optimal sets of projections via training a machine learning  (ML) model.
To this end,
we first derive the soft-decision based version of the subRPA algorithm, called ``soft-subRPA'', that not only improves upon the performance of the subRPA algorithm but also provides a differentiable version of our decoding algorithm.
 Enabled by our differentiable soft-subRPA algorithm, we train an ML model to search for the \textit{good} sets of projections. 
 We find out that carefully training our ML model provides the possibility to find the best sets of projections that achieve very close to the performance of full-projection decoding with much smaller number of projections. 
 
 We would like to highlight that our work also adds to the rich literature on soft-decision decoding of algebraic codes, including the celebrated work by Koetter and Vardy on soft-decision decoding of Reed-Solomon codes \cite{koetter2003algebraic}, which is also used for soft-decision decoding of other algebraic codes such as Hermitian codes \cite{lee2010algebraic} and elliptic codes \cite{wan2022algebraic}, as well as the work by Vardy and Be'ery on soft-decision decoding of Bose–Chaudhuri–Hocquenghem (BCH) codes \cite{vardy1994maximum}, among others.

 Finally, besides designing efficient decoding algorithms, we also provide some insights on encoding RM subcodes by empirically investigating their performance.
Our results show that constructing the code generator matrix with respect to a lower complexity for our algorithms results in a superior performance compared to a higher complexity generator matrix. Also, our empirical results for pruning projections mostly suggest a superior performance for the projection sets incurring a lower decoding complexity. This together with our observation on the encoding part unravels a two-fold gain for our proposed algorithms: {a better performance for a lower complexity}.

The rest of the paper is organized as follows. In Section \ref{Prelim}, we provide some preliminaries on RM codes and RPA decoding. In Section \ref{eff_dec}, we present the subRPA and soft-subRPA algorithms for decoding RM subcodes. We empirically investigate encoding of RM subcodes and present several ad-hoc projection pruning schemes in Section \ref{sec_insights}. 
Section \ref{sec_nn} is devoted to our ML-aided projection pruning algorithm, 
and Section \ref{conc} concludes the paper.

\section{Preliminaries}\label{Prelim}
In this section, we briefly review RM codes
 and the RPA algorithm. The reader is referred to \cite{ye2020recursive} for additional details on the RPA algorithm.
 
\subsection{RM Codes}\label{RM_rev}
Let $k$ and $n$ denote the code dimension and blocklength, respectively. Also, let $m=\log_2 n$. The $r$-th order RM code of length $2^m$, denoted by $\mathcal{RM}(m,r)$, is then defined by the following set of vectors as the basis
\begin{align}\label{rm_basis}
\{\boldsymbol{v}_m(\mathcal{A}):~\mathcal{A}\subseteq[m],|\mathcal{A}|\le r\},
\end{align}
where $[m]:=\{1,2,\dots,m\}$, $|\mathcal{A}|$ denotes the size of the set $\mathcal{A}$, and $\boldsymbol{v}_m(\mathcal{A})$ is a row vector of length $2^m$ whose components are indexed by binary vectors $\boldsymbol{z}=(z_1,z_2,\dots,z_m) \in \{0,1\}^m$ 
 as 
\begin{align}\label{v_mA}
\boldsymbol{v}_m(\mathcal{A},\boldsymbol{z}) = \prod_{i\in \mathcal{A}} z_i,
\end{align}
with the convention of $\prod_{i\in \mathcal{\emptyset}} z_i:=1$.
It follows
  from \eqref{rm_basis} that $\mathcal{RM}(m,r)$ has the dimension of 
\begin{align}\label{k}
k=\sum_{i=0}^r \binom{m}{i}.
\end{align}

Given the basis in
\eqref{rm_basis}, the (codebook of) $\mathcal{RM}(m,r)$ code is defined as the following set of binary vectors
\begin{align}\label{rm_codebook}
\mathcal{RM}(m,r):= \left\{\sum_{\mathcal{A}\subseteq[m],|\mathcal{A}|\le r}\hspace{-0.5cm}u(\mathcal{A}) \boldsymbol{v}_m(\mathcal{A}): u(\mathcal{A})\in\{0,1\}~
\forall \mathcal{A}\right\}.
\end{align}
Therefore, considering a polynomial ring $\mathbb{F}_2[Z_1,Z_2,\dots,Z_m]$ of $m$ variables, the components of $\boldsymbol{v}_m(A)$ are the evaluations of the monomial $\prod_{i\in \mathcal{A}}Z_i$ at points $\boldsymbol{z}$ in the vector space $\mathbb{E}:=\mathbb{F}_2^m$. Moreover, each codeword $\boldsymbol{c}=(\boldsymbol{c}(\boldsymbol{z}), \boldsymbol{z}\in\mathbb{E})\in\mathcal{RM}(m,r)$, that is also indexed by the binary vectors $\boldsymbol{z}$, is defined as the evaluations of an $m$-variate polynomial with degree at most $r$ at points $\boldsymbol{z}\in\mathbb{E}$.

\subsection{RPA Decoding Algorithm}\label{RPA_rev}
The RPA algorithm is comprised of the following three
main phases.
\subsubsection{Projection} The RPA algorithm starts by projecting the received corrupted binary vector (in the case of BSC) or the log-likelihood ratio (LLR) vector of the channel output (in the case of general binary-input memoryless channels) onto the subspaces of $\mathbb{E}$. Considering $\mathbb{B}$ as a $s$-dimensional subspace of $\mathbb{E}$, with $s\leq r$, the quotient space $\mathbb{E}/\mathbb{B}$ contains all the cosets of $\mathbb{B}$ in $\mathbb{E}$. Each coset $\mathcal{T}\in \mathbb{E}/\mathbb{B}$ has the form $\mathcal{T}=\boldsymbol{z}+\mathbb{B}$ for some $\boldsymbol{z}\in\mathbb{E}$. Then, in the case of BSC, the projection of the channel binary output $\boldsymbol{y}=(\boldsymbol{y}(\boldsymbol{z}), \boldsymbol{z}\in\mathbb{E})$ onto the cosets of $\mathbb{B}$ is defined as
\begin{align}  \label{proj_bsc}
\boldsymbol{y}_{/\mathbb{B}} :=\big(\boldsymbol{ y}_{/\mathbb{B}}(\mathcal{T}), \mathcal{T}\in\mathbb{E}/\mathbb{B} \big), \text{~s.t.~~} \boldsymbol{ y}_{/\mathbb{B}}(\mathcal{T}) := \bigoplus_{\boldsymbol{z}\in \mathcal{T}} \boldsymbol{ y}(\boldsymbol{z}),
\end{align}
where $\bigoplus$ denotes the coordinate-wise addition in $\mathbb{F}_2$. For the binary-input memoryless channels the RPA algorithm works on the projection of the channel output LLR vector $\boldsymbol{l}$. In the case of a one-dimensional subspace $\mathbb{B}$, the projected LLR vector can be obtained as $\boldsymbol{l}_{/\mathbb{B}} :=(\boldsymbol{l}_{/\mathbb{B}}(\mathcal{T}), \mathcal{T}\in\mathbb{E}/\mathbb{B})$, where
\begin{align}  \label{proj_awgn}
\boldsymbol{l}_{/\mathbb{B}}(\mathcal{T})\!=\!\ln\!\Big(\!\exp \big(\sum_{\boldsymbol{z}\in\mathcal{T}} \boldsymbol{l}(\boldsymbol{z})\big)\!+\!1\Big)\!-\!
\ln \Big( \sum_{\boldsymbol{z}\in \mathcal{T}} \exp(\boldsymbol{l}(\boldsymbol{z})) \Big).
\end{align}
In the case of a general $s$-dimensional subspace $\mathbb{B}$, the quotient space $\mathbb{E}/\mathbb{B}$ contains $2^{m-s}$ cosets $\mathcal{T}$ each of size $2^s$. Then, one can follow a similar approach to the proof of \cite[Eq. (13)]{ye2020recursive} to prove that $\boldsymbol{l}_{/\mathbb{B}}(\mathcal{T})$, for each coset $\mathcal{T}$, can be obtained recursively as
\begin{align}\label{rec_proj}
\boldsymbol{l}_{/\mathbb{B}}(\mathcal{T})=\ln\left(\frac{1+\exp \left(\boldsymbol{l}_{/\mathbb{B}}(\mathcal{T}_{1:2^{s-1}}) + \boldsymbol{l}_{/\mathbb{B}}(\mathcal{T}_{1+2^{s-1}:2^s})  \right)} 
{\exp\left(\boldsymbol{l}_{/\mathbb{B}}(\mathcal{T}_{1:2^{s-1}})\right) + \exp\left(\boldsymbol{l}_{/\mathbb{B}}(\mathcal{T}_{1+2^{s-1}:2^s})\right) }\right),
\end{align}
where the notation $\mathcal{T}_{i:j}$ is used to denote the subset of $\mathcal{T}$ containing all the elements from index $i$ to $j$. For the base case of the recursive equation \eqref{rec_proj} one can use $s=1$ to obtain \eqref{proj_awgn} as the base case. Alternatively, we can set $s=0$ as the base case with the convention of $\boldsymbol{l}_{/\mathbb{B}}(\mathcal{T}):=\boldsymbol{l}(\boldsymbol{z})$ for a set $\mathcal{T}$ containing a single element $\boldsymbol{z}$. In the latter case, we can derive \eqref{proj_awgn} as a special case of \eqref{rec_proj} by setting $s=1$.

\subsubsection{Decoding the Projected Outputs}
Once the decoder projects the channel output ($\boldsymbol{y}$ or $\boldsymbol{l}$), it starts recursively decoding the projected outputs, i.e., it projects them onto new subspaces and continues until the projected outputs correspond to order-$1$ RM codes. The decoder then applies the fast Hadamard transform (FHT) \cite{macwilliams1977theory} to efficiently decode order-$1$ codes. By using the FHT algorithm, one can implement the MAP decoder for the first-order RM codes with the complexity $\mathcal{O}(n\log n)$ instead of $\mathcal{O}(n^2)$. Once the first-order codes are decoded, the algorithm \textit{aggregates} the outputs (as explained next) to decode the codes at a higher layer. The decoder may also iterate the whole process, at each middle decoding step, several times to ensure the convergence of the algorithm.
\subsubsection{Aggregation} At each layer in the decoding process (and each node in the decoding tree), the decoder needs to \textit{aggregate} the output of the channel at that node with the decoding results of the next (underneath) layer to update the channel output. Note that the channel output at a given node can be either the actual channel output ($\boldsymbol{y}$ or $\boldsymbol{l}$) or the projected ones, depending on the depth of that node in the decoding tree of the recursive algorithm. Several aggregation algorithms are presented in \cite{ye2020recursive} for one- and two-dimensional subspaces. We refer the reader to \cite{ye2020recursive} for the details on the  aggregation methods.


\section{Efficient Decoding of RM Subcodes}\label{eff_dec}
\subsection{Problem Setting}\label{settings}
An equivalent description of the RM encoder can be obtained through the so-called polarization matrix. Indeed, the generator matrix of an $\mathcal{RM}(m,r)$ code, denoted by $\boldsymbol{G}_{k\times n}$, can be obtained by choosing rows of the following matrix that have a Hamming weight of at least $2^{m-r}$:
\begin{align}\label{Pnn}
\boldsymbol{P}_{n\times n}=\begin{bmatrix}
1 & 0\\ 1&1
\end{bmatrix}^{\otimes m},
\end{align}
where $\boldsymbol{F}^{\otimes m}$ is the $m$-th Kronecker power of a matrix $\boldsymbol{F}$. The resulting generator matrix $\boldsymbol{G}_{k\times n}$ can then be partitioned into sub-matrices as
\begin{align}\label{Gkn}
\boldsymbol{G}_{k\times n}=\begin{bmatrix}
\boldsymbol{G}_{0}\\
\boldsymbol{G}_{1}\vspace{-0.1cm}\\
\vdots\\
\boldsymbol{G}_{r-1}\\
\boldsymbol{G}_{r}
\end{bmatrix},
\end{align}
where $\boldsymbol{G}_{0}$ is a length-$n$ all-one row vector, and $\boldsymbol{G}_{1}$ is an $m\times n$ matrix that lists all the $n=2^m$ unique length-$m$ binary vectors $\{0,1\}^m$ as the columns. Moreover, $\boldsymbol{G}_{i}$, for $1\leq i\leq r$, is an $\binom{m}{i}\times n$ matrix whose each row is obtained by the element-wise product of a distinct selection of $i$ rows from $\boldsymbol{G}_{1}$ \cite{salomon2005augmented}. Accordingly, $\boldsymbol{G}_{k\times n}$ has exactly $\binom{m}{i}$ rows with the Hamming weight $n/2^i$, for $0\leq i\leq r$.


As seen, the RM encoder does not allow choosing any desired code dimension; it should be of the form $k=\sum_{i=0}^{r}\binom{m}{i}$ for some $r\in\{0,1,\cdots,m\}$. Suppose that we want to construct a subcode of $\mathcal{RM}(m,r)$ with a dimension $k$ such that $k_l< k< k_u$, where $k_l :=\sum_{i=0}^{r-1}\binom{m}{i}$ and $k_u :=\sum_{i=0}^{r}\binom{m}{i}$ for some $r\in[m]$. Given that the construction of RM codes corresponds to picking rows of $\boldsymbol{P}_{n\times n}$ that have the highest Hamming weights, the first $k_l$ rows of the generator matrix $\boldsymbol{G}_{k\times n}$ will be the same as the generator matrix of the lower-order RM code, i.e.,  $\mathcal{RM}(m,r-1)$, that has a Hamming weight of at least $2^{m-r+1}$. It then remains to pick extra $k-k_l$ rows from $\boldsymbol{P}_{n\times n}$. These will be picked from the additional $k_u-k_l=\binom{m}{r}$ rows in 
$\boldsymbol{G}_{r}$ since they all have the same Hamming weight of $2^{m-r}$, which is the next largest Hamming weight. In a sense, we limit our attention to RM subcodes that, roughly speaking, \textit{sit} between two RM codes of consecutive orders. More specifically, they are subcodes of $\mathcal{RM}(m,r)$ and also contain $\mathcal{RM}(m,r-1)$ as a subcode, for some $r \in [m]$. The question is then how to choose the extra $k-k_l$ rows out of those $\binom{m}{r}$ rows of weight $2^{m-r}$ to construct an RM subcode of dimension $k$ as specified above.
This important question requires
 a separate follow-up work and is beyond the scope of this paper. In the meantime, we provide some insights regarding the encoding of RM subcodes in Section \ref{sec_encoding} after describing our decoding algorithms in Sections \ref{sec_subRPA} and \ref{sec_Soft-subRPA} with respect to a generic generator matrix $\boldsymbol{G}_{k\times n}$.
Our results show that randomly selecting a subset of those rows is not always good. Indeed, some selections are better that the others, and also the set of \textit{good} rows can depend on the underlying decoding algorithm.

\subsection{SubRPA Decoding Algorithm}\label{sec_subRPA}
Before delving into the description of our decoding algorithms, we first need to emphasize some important facts.

\noindent{\textbf{Remark 1.}} The result of the projection operation corresponds to a code with the generator matrix that is formed by merging (i.e., binary addition of) the columns of the original code generator matrix indexed by the cosets of the projection subspace. This is clear for the BSC model, as formulated in \eqref{proj_bsc}. Additionally, for general BMS channels, the objective is to estimate the projected codewords $\boldsymbol{ c}_{/\mathbb{B}}(\mathcal{T})$'s, $\mathcal{T}\in\mathbb{E}/\mathbb{B}$, based on the channel (projected) LLRs \cite{ye2020recursive}; hence, the same principle follows for any BMS channels.


\begin{proposition}\label{prop_subcode}
	Let $\mathcal{C}$ be a subcode of $\mathcal{RM}(m,r)$ with dimension $k$ such that $k_l< k< k_u$, where $k_l :=\sum_{i=0}^{r-1}\binom{m}{i}$ and $k_u :=\sum_{i=0}^{r}\binom{m}{i}$ for some $r\in[m]$. The projection of this code onto $s$-dimensional subspaces of $\mathbb{E}$, $1\leq s\leq r-1$, results in subcodes of $\mathcal{RM}(m-s,r-s)$. It is also possible for the projected codes to be $\mathcal{RM}(m-s,r-s)$ or $\mathcal{RM}(m-s,r-1-s)$ codes.
\end{proposition}
\begin{proof}
	Please refer to Appendix \ref{app_prop1}.
\end{proof}

Hereafter, for the sake of brevity, we simply say that the projections of a subcode of $\mathcal{RM}(m,r)$ code onto the $s$-dimensional subspaces of $\mathbb{E}$ are subcodes of $\mathcal{RM}(m-s,r-s)$; however, we still mean the precise statement in \Pref{prop_subcode}. Now, we are ready to present our decoding algorithms for RM subcodes. Our algorithms are based on projecting onto one-dimensional (1-D) subspaces. However, they can be generalized to the case of $s$-dimensional subspaces by following a similar approach.


As schematically shown in Fig. \ref{fig0_1}, the subRPA algorithm proceeds in a similar way to the RPA algorithm. More precisely, it first projects the code $\mathcal{C}$, that is a subcode of $\mathcal{RM}(m,r)$, onto 1-D subspaces to get subcodes of $\mathcal{RM}(m-1,r-1)$ at the next layer. It then recursively applies the subRPA algorithm to decode these projected codes. Next, it aggregates the decoding results of the  next layer with the output LLRs of the current layer (similar to \cite[Algorithm 4]{ye2020recursive}) to update the LLRs. Finally, it iterates this process several times to ensure the convergence of the algorithm, and takes the sign of the updated LLRs to obtain the decoded codewords.

The main distinction between the subRPA and RPA algorithms, however, is the decoding of the projected codes at the bottom layer. Based on \Pref{prop_subcode}, after $r-1$ layers of 1-D projections, the decoder ends up with subcodes of $\mathcal{RM}(m-r+1,1)$ at the bottom layer. These projected codes can have different dimensions though all are less than or equal to $m-r+2$. Therefore, the subRPA algorithm, manageably, applies the MAP decoding at the bottom layer. 

Given that the projected codewords at the bottom layer are not all from the same code, the MAP decoding should be carefully performed. Based on Remark 1, the projected codes at the bottom layer can be obtained from the so-called \textit{projected generator matrices} of dimension $k\times 2^{m-r+1}$, after $r-1$ times (binary) merging of the $2^m$ columns of the original generator matrix $\boldsymbol{G}_{k\times n}$. However, many of these $k$ rows of the projected generator matrices are linearly dependent. In fact, all of these matrices have ranks (i.e., code dimensions) of less than or equal to $m-r+2$. In order to facilitate the MAP decoding at the bottom layer, we can pre-compute and store the codebook of each projected code at the bottom layer. Particularly, let $R_{t}$ be the rank of the $t$-th projected generator matrix $\boldsymbol{G}_p^{(t)}$ at the bottom layer, $t\in[T]$, where $T$ is the total number of projected codes at the bottom layer (which depends on the number of layers as well as the number of projections per layer). We can then pre-compute the codebook $\mathcal{C}_p^{(t)}$ that contains the $2^{R_{t}}$ length-($n/2^{r-1}$) codewords $\boldsymbol{c}^{(t)}_{p,i_t}$, $i_t\in[2^{R_{t}}]$, of the $t$-th projected code at the bottom layer. Now, given the projected LLR vector $\boldsymbol{l}^{(t)}_p$ of length $n/2^{r-1}$ at the bottom layer, we pick the codeword $\boldsymbol{c}^{(t)}_{p,i^*}$ that maximizes the MAP rule for BMS channels \cite{ye2020recursive}, i.e., 
\begin{align}\label{map}
\hat{\boldsymbol{y}}_{t}=\boldsymbol{c}^{(t)}_{p,i^*},~~ \text{s.t.} ~~~~ i^*=\operatorname*{argmax}_{i_t\in[2^{R_{t}}]}~~ \langle\boldsymbol{l}^{(t)}_p,1-2{\boldsymbol{c}^{(t)}_{p,i_t}}\rangle,
\end{align}
where $\langle \cdot,\cdot\rangle$ denotes the inner (dot) product of two vectors.
An efficient algorithm for computing $\mathcal{C}_p^{(t)}$ given $\boldsymbol{G}_p^{(t)}$ is presented in Algorithm \ref{U_Cp_Alg} in Section \ref{sec_Soft-subRPA}.

\begin{figure}[t]
	\centering
	\includegraphics[trim=0.1cm 0.2cm 0.1 0.1cm,width=3.4in]{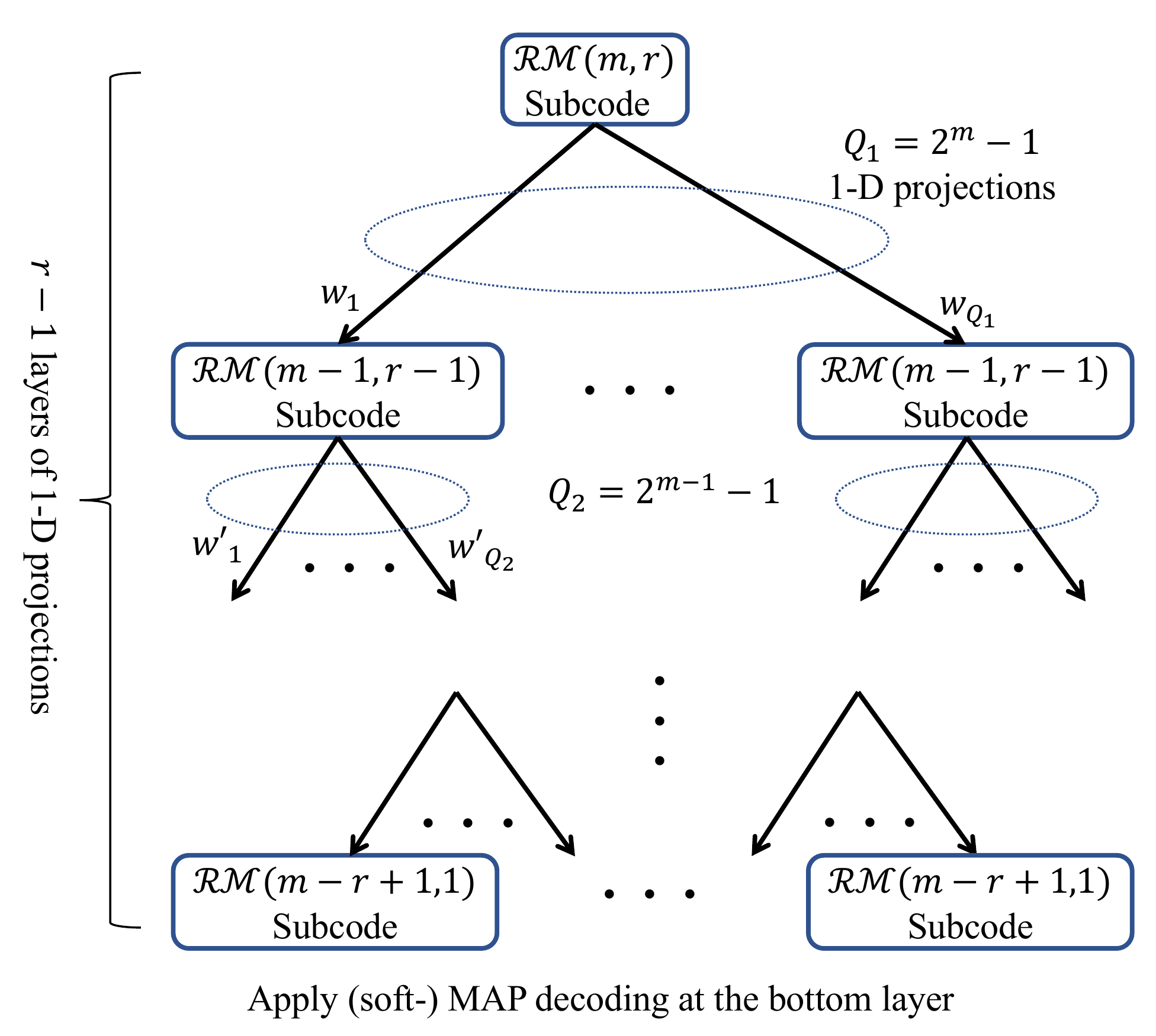}
	\caption{Schematic diagram of the subRPA and soft-subRPA algorithms.
	 }
	\label{fig0_1}
\end{figure}

\subsection{Soft-SubRPA Algorithm}\label{sec_Soft-subRPA}
In this section, we derive the soft-decision version of the subRPA algorithm, referred to as \textit{soft-subRPA} in this paper. As schematically shown in Fig. \ref{fig0_1}, the soft-subRPA algorithm obtains soft decisions at the bottom layer instead of performing hard MAP decodings; this process is called \textit{soft-MAP} in this paper. Additionally, the decoder applies a different rule to aggregate the soft decisions obtained from the next layers with the LLRs available at the current layer; we refer to this aggregation process as \textit{soft-aggregation}. The soft-subRPA algorithm not only improves upon the performance of the subRPA but also replaces the hard MAP decodings at the bottom layer with a differentiable operation that, in turn, enables training an ML model as delineated in Section \ref{sec_nn}.

The soft-MAP algorithm for making soft decisions on the projected codes at the bottom layer, that are subcodes of first-order RM codes, is presented in Algorithm \ref{alg_softmap} for the case of the additive white Gaussian noise (AWGN) channel. The process is comprised of two main steps : 1) obtaining the LLRs of the information bits, and 2) obtaining the soft decisions (i.e., LLRs) of the coded bits using that of information bits. Note that we invoke \textit{max-log} and \textit{min-sum} approximations, to be clarified later, in Algorithm \ref{alg_softmap}. For the sake of brevity, let us drop the superscript $t$. Particularly, let $R$ be the rank of the projected generator matrix $\boldsymbol{G}_p$ of a projected code at the bottom layer with codebook $\mathcal{C}_p$. Also, assume a $2^R\times k$ matrix $\boldsymbol{U}$ that lists all $2^R$ length-$k$ sequences of bits that generate the codebook $\mathcal{C}_p$ (through modulo-$2$ matrix multiplication $\boldsymbol{U}\boldsymbol{G}_p$). 

An efficient algorithm for computing matrix $\boldsymbol{U}$ and codebook $\mathcal{C}_p$ for a given projected generator matrix $\boldsymbol{G}_p$ is presented in Algorithm \ref{U_Cp_Alg}. In Algorithm \ref{U_Cp_Alg}, $\texttt{gfrank}(\boldsymbol{A},2)$ is a function that computes the rank of the matrix $\boldsymbol{A}$ over the binary field. Moreover, $\texttt{de2bi}(a:b,m)$ is a function that outputs a $(b-a+1)\times m$ matrix whose rows are the length-$m$ binary representations of all the integers from $a$ to $b$.
The algorithm first iterates over the rows of $\boldsymbol{G}_p$ to find the index of the (first) $R$ linearly independent rows, i.e., the index of the rows forming a basis for $\boldsymbol{G}_p$. The algorithm stops iterating over the remaining rows as soon as $R$ linearly independent rows are found (i.e., when $r=R$) to avoid unnecessary work. Once the set $\mathcal{U}_{\text{ind}}$ of those indices is found, the $2^R\times k$ matrix $\boldsymbol{U}$ is formed by inserting all distinct binary vectors of length $R$ in the $R$ columns of $\boldsymbol{U}$ indexed by the set $\mathcal{U}_{\text{ind}}$, and freezing the remaining $k-R$ columns to zero. Finally, the codebook $\mathcal{C}_p$ is obtained by the matrix multiplication of $\boldsymbol{U}\boldsymbol{G}_p$ over $\mathbb{F}_2$. { The memory required to store the projected generator matrices and codebooks at the bottom layer is quantified in Appendix \ref{app_memory}.}

Given that only $R$ indices of the length-$k$ sequences in $\boldsymbol{U}$ contain the information bits (and the remaining bit positions are frozen to $0$), the objective of the first step of the soft-MAP algorithm is to obtain the LLRs of the $R$ information bits using the available projected LLR vector $\boldsymbol{l}_p$. This can be done, using \eqref{llr_inf} in Appendix \ref{appnd_LLR_inf} invoking max-log approximation, as described in Algorithm \ref{alg_softmap}. Note that the LLRs of the $k-R$ indices that do not carry information are set to zero. 

\begin{algorithm}[t]
	\caption{Soft-MAP Algorithm for the AWGN Channel} \label{alg_softmap}
	\textbf{Input:} The LLR vector $\boldsymbol{l}_p$; the generator matrix $\boldsymbol{G}_p$; the codebook $\mathcal{C}_p$; and the matrix $\boldsymbol{U}$ of the information sequences
	
	\textbf{Output:} Soft decisions (i.e., the updated LLR vector) $\hat{\boldsymbol{l}}$
	\vspace*{0.05in}
	\begin{algorithmic}[1]
		\State $k\gets$ number of rows in $\boldsymbol{G}_p$
		\State  $\boldsymbol{l}_{\rm inf}\gets \mathbf{0}_k$ \Comment initialize $\boldsymbol{l}_{\rm inf}$ as a length-$k$ all-zero vector
		\State $\boldsymbol{\tilde{C}}\gets 1-2{\boldsymbol{C}}$ \Comment $\boldsymbol{C}$ is the codebook matrix (in binary)
		\State $\boldsymbol{\tilde{l}}\gets \boldsymbol{l}_p\boldsymbol{\tilde{C}}^{T}$ \Comment matrix mul. of $\boldsymbol{l}_p$ with the transpose of $\boldsymbol{\tilde{C}}$
		\For {$i=1,2,\cdots,k$} \Comment obtaining inf. bits LLRs
		\If {$\boldsymbol{U}(:,i)\neq\mathbf{0}$ ($i$-th column is not frozen to $0$)}
		\State $\displaystyle\boldsymbol{l}_{\rm inf}(i) \gets \operatorname*{max}_{i'\in\{i':\boldsymbol{U}(i'\!,i)=0\}}\!\boldsymbol{\tilde{l}}(i')~-\hspace{-0.1cm}\operatorname*{max}_{i'\in\{i':\boldsymbol{U}(i'\!,i)=1\}}\!\boldsymbol{\tilde{l}}(i')$
		\EndIf
		\EndFor 
		\State $n'\gets$ number of columns in $\boldsymbol{G}_p$
		\State  $\boldsymbol{l}_{\rm enc}\gets \mathbf{0}_{n'}$ \Comment initialize $\boldsymbol{l}_{\rm enc}$ as a length-$n'$ all-zero vector
		
		\State  Initialize $\boldsymbol{l}_{\rm enc}$ as an all-zero vector of length $n'$
		\State $\boldsymbol{L}\gets \texttt{repeat}(\boldsymbol{l}_{\rm inf}^T,1,n')$ \Comment  make $n'$ copies of $\boldsymbol{l}_{\rm inf}^T$
		\State $\boldsymbol{V}\gets \boldsymbol{L}\odot \boldsymbol{G}_p$ \Comment element-wise matrix multiplication
		\For {$j=1,2,\cdots,n'$}
		\State $\boldsymbol{v}\gets$ vector containing nonzero elements of $\boldsymbol{V}(:,j)$
		\State $\boldsymbol{l}_{\rm enc}(j)\gets\prod_{j'}{\rm sign}(\boldsymbol{v}(j')) \times \operatorname*{min}_{j'} |\boldsymbol{v}(j')|$
		\EndFor 
		\State $\hat{\boldsymbol{l}}\gets\boldsymbol{l}_{\rm enc}$
		\State \textbf{return} $\hat{\boldsymbol{l}}$
	\end{algorithmic}
\end{algorithm}

\begin{algorithm}[t]
	\caption{Matrix $\boldsymbol{U}$ and codebook $\mathcal{C}_p$ Finder}    \label{U_Cp_Alg}
	\textbf{Input:} The  projected generator matrix $\boldsymbol{G}_p$
	
	\textbf{Output:} Matrix $\boldsymbol{U}$ of the information sequences; and codebook $\mathcal{C}_p$ of the projected code
	\vspace*{0.05in}
	\begin{algorithmic}[1]
		
		\State $k\gets$ number of rows in $\boldsymbol{G}_p$
		\State $\mathcal{U}_{\text{ind}}\gets\{\}$ \Comment initialize $\mathcal{U}_{\text{ind}}$ as an empty set
		\State $r\gets 0$
		\State $\boldsymbol{G}^{\text{tmp}}_p\gets [~]$ \Comment initialize $\boldsymbol{G}^{\text{tmp}}_p$ as an empty matrix
		\State $R\gets \texttt{gfrank}(\boldsymbol{G}_p,2)$
		\State $i\gets 1$
		\While {$i\leq k$ \textbf{and} $r<R$} \Comment iterate over the rows of $\boldsymbol{G}_p$
		\State Add the $i$-th row of $\boldsymbol{G}_p$ to $\boldsymbol{G}^{\text{tmp}}_p$
		\State $i\gets i+1$
		\If {$\texttt{gfrank}(\boldsymbol{G}^{\text{tmp}}_p,2)>r$}
		\State $r\gets r+1$
		\State Add $i$ to $\mathcal{U}_{\text{ind}}$
		\EndIf
		\EndWhile
		
		\State $\boldsymbol{U}\gets \mathbf{0}_{2^R\times k}$ \Comment initialize $\boldsymbol{U}$ as an all-zero $2^R\times k$ matrix
		\State $\boldsymbol{U}(:,\mathcal{U}_{\text{ind}})\gets \texttt{de2bi}(0:2^R-1,R)$ \Comment fill out
		the columns in $\boldsymbol{U}$ indexed by the set $\mathcal{U}_{\text{ind}}$ with the $2^R$ distinct  binary vectors of length $R$
		\State $\boldsymbol{C}\gets\boldsymbol{U}\boldsymbol{G}_p~\texttt{mod}~2$ \Comment matrix multiplication over $\mathbb{F}_2$
		\State $\mathcal{C}_p\gets$ rows of $\boldsymbol{C}$ \Comment list all rows of $\boldsymbol{C}$ in $\mathcal{C}_p$
		\State \textbf{return} $\boldsymbol{U}$ and $\mathcal{C}_p$
	\end{algorithmic}
\end{algorithm}

Once the LLRs of the information bits are calculated, they can be combined according to the columns of $\boldsymbol{G}_p$ to obtain the LLRs of the encoded bits $\boldsymbol{l}_{\rm enc}$. The codewords in $\mathcal{C}_p$ are obtained by the multiplication of $\boldsymbol{U}\boldsymbol{G}_p$, i.e., each $j$-th coded bit, $j\in[n']$, where $n'$ is the code length, is obtained based on the linear combination of the information bits $u_i$'s according to the $j$-th column of $\boldsymbol{G}_p$. Therefore, we can apply the well-known min-sum approximation to calculate the LLR vector of the coded bits as $\boldsymbol{l}_{\rm enc}:=(\boldsymbol{l}_{\rm enc}(j), j\in[n'])$, where
\begin{align}\label{llr_enc}
\boldsymbol{l}_{\rm enc}(j)=\prod_{i\in \Delta_{j}}{\rm sign}(\boldsymbol{l}_{\rm inf}(i)) \times \operatorname*{min}_{i\in \Delta_{j}} |\boldsymbol{l}_{\rm inf}(i)|,
\end{align}
where $\Delta_{j}$ is the set of indices defining the nonzero elements in the element-wise multiplication of $\boldsymbol{l}_{\rm inf}$ (to skip the frozen bit positions under the formulation of this paper) with the $j$-th column of $\boldsymbol{G}_p$. This process is summarized in Algorithm \ref{alg_softmap} in an efficient way. The decoder may also iterate the whole process several times to ensure the convergence of the soft-MAP algorithm.

Finally, given the soft decisions at the bottom layer, the decoder needs to aggregate the decisions with the current LLRs. 
In the following, we first define the ``soft-aggregation'' scheme as an extension of the aggregation method in \cite[Algorithm 4]{ye2020recursive} for the case of soft decisions.

\begin{definition}[Soft-Aggregation]\label{def_softaggr}
	Let $\boldsymbol{l}$ be the vector of the channel LLRs, with length $n=2^m$, at a given layer. Suppose that there are $Q$ 1-D subspaces $\mathbb{B}_q$, $q\in[Q]$, to project this LLR vector at the next layer (in the case of full-projection decoding, there are $n-1$ 1-D subspaces, hence $Q=n-1$). 
	Also, let $\boldsymbol{\hat{l}}_q$ denote the length-$n/2$ vector of soft decisions of the projected LLRs according to Algorithm \ref{alg_softmap}. The ``soft-aggregation'' of $\boldsymbol{l}$ and $\boldsymbol{\hat{l}}_q$'s is defined as a length-$n$ vector $\tilde{\boldsymbol{l}}:=(\tilde{\boldsymbol{l}}(\boldsymbol{z}), \boldsymbol{z}\in\mathbb{F}_2^m)$ where 
	\begin{align}\label{eq_softaggr}
	\tilde{\boldsymbol{l}}(\boldsymbol{z})=\frac{1}{Q}\sum_{q=1}^{Q}\tanh\big(\boldsymbol{\hat{l}}_q\left([\boldsymbol{z}+\mathbb{B}_q]\right)/2\big) \boldsymbol{l}(\boldsymbol{z}\oplus\boldsymbol{z}_q).
	\end{align}
	where $\boldsymbol{z}_q$ is the nonzero vector of the 1-D subspace $\mathbb{B}_q$, and $[\boldsymbol{z}+\mathbb{B}_q]$ is the coset containing $\boldsymbol{z}$ for the projection onto $\mathbb{B}_q$. 
\end{definition}
In order to observe \eqref{eq_softaggr}, recall that the objective of the aggregation step is to update the length-$n$ channel LLR vector $\boldsymbol{l}$ to $\tilde{\boldsymbol{l}}$ given the soft decisions of the projected codes.  $\boldsymbol{\hat{l}}_q\left([\boldsymbol{z}+\mathbb{B}_q]\right)$ {serves} as a soft estimate of the binary addition of the coded bits at positions $\boldsymbol{z}$ and $\boldsymbol{z}\oplus\boldsymbol{z}_q$. Hence, by following similar arguments to \cite{ye2020recursive}, if that combined bit is $0$, then the updated LLR at position $\boldsymbol{z}$ should take the same sign as the channel LLR at position $\boldsymbol{z}\oplus \boldsymbol{z}_q$. Note that this happens with probability $a_0:=1/\big[1+\exp\big(-\boldsymbol{\hat{l}}_q\left([\boldsymbol{z}+\mathbb{B}_q]\right)\big)\big]$. Similarly, with probability $a_1:=1/\big[1+\exp\big(\boldsymbol{\hat{l}}_q\left([\boldsymbol{z}+\mathbb{B}_q]\right)\big)\big]$ the combined bit is $1$, and hence the updated LLR at position $\boldsymbol{z}$ and $\boldsymbol{l}(\boldsymbol{z}\oplus\boldsymbol{z}_q)$ should have different signs. Therefore, given a projection subspace $\mathbb{B}_q$, one can update the channel LLR as $a_0\times\boldsymbol{l}(\boldsymbol{z}\oplus\boldsymbol{z}_q)+a_1\times-\boldsymbol{l}(\boldsymbol{z}\oplus\boldsymbol{z}_q)$. Taking the average over all $Q$ projections then results in the soft-aggregation rule in \eqref{eq_softaggr}.

It is worth mentioning that one can also 
 update the channel LLR as
\begin{align}\label{eq_logsum}
\tilde{\boldsymbol{l}}_{\rm ls}(\boldsymbol{z})=\frac{1}{Q}\sum_{q=1}^{Q}\ln\left(\frac{1+{\rm e}^{\boldsymbol{\hat{l}}_q\left([\boldsymbol{z}+\mathbb{B}_q]\right)+\boldsymbol{l}(\boldsymbol{z}\oplus\boldsymbol{z}_q)}}{{\rm e}^{\boldsymbol{\hat{l}}_q\left([\boldsymbol{z}+\mathbb{B}_q]\right)}+{\rm e}^{\boldsymbol{l}(\boldsymbol{z}\oplus\boldsymbol{z}_q)}}\right).
\end{align}
The rationale behind \eqref{eq_logsum} follows by similar arguments as above and then deriving the LLR of the sum of two binary random variables given the LLRs of each of them. Therefore, \eqref{eq_logsum} is an exact expression assuming  independence among the involved LLR components.   Our empirical observations, however, suggest almost identical results for either aggregation methods. Therefore, given the complexity of computing expressions like \eqref{eq_logsum}, one can reliably apply our proposed soft-aggregation method in Definition \ref{def_softaggr}.

{\noindent{\textbf{Remark 2.}} The subRPA and soft-subRPA decoding algorithms reduce to the original RPA decoding algorithm \cite{ye2020recursive} and its soft version, respectively, when applied to an RM code instead of an RM subcode (i.e., when the code dimension $k$, for a given $m$, follows Eq. \eqref{k}). The only difference is the decoding at the bottom layer, where the FHT decoding can then be directly applied given that all projected codes are order-1 RM codes. Therefore, the proposed ML training approach in Section \ref{sec_nn} can be readily applied to the RM codes as well. However, we will \textit{empirically} establish (see Fig. \ref{fig10}) that the performance of a pruned-projection decoding of an RM code is (almost) the same regardless of the selection of the projections. Therefore, not much (if any) gain can be expected from ML training for projection selection in RPA decoding of RM codes, and simply a random selection of the projections may be sufficient for RPA decoding of RM codes.}

{ Before concluding this section, in the following proposition, we characterize the complexity of our proposed decoding algorithms under different settings
\begin{proposition}\label{prop_comp}
	The decoding complexity of our proposed (soft-) subRPA algorithm in decoding a subcode of an $\mathcal{RM}(m,r)$ code, $r>1$,  is $\mathcal{O}(n^{r-1} \mathcal{C}(m-r+1,1))$, where $\mathcal{C}(m',1)$ stands for the complexity of decoding a subcode of an $\mathcal{RM}(m',1)$ code. Assuming (soft-) MAP at the bottom layer, $\mathcal{C}(m-r+1,1)=\mathcal{O}(n^2/2^{2r-3})$, and the overall decoding complexity simplifies to $\mathcal{O}(n^{r+1})$. The decoding complexity reduces to $\mathcal{O}(n^2)$ for pruned-projection decoding with factor $\beta=\mathcal{O}(1/n)$. The overall complexity further reduces to $\mathcal{O}(n)$ if $2^{R_t}=\mathcal{O}(1)$, $\forall t\in[T]$, in addition to $\beta=\mathcal{O}(1/n)$, where $R_t$ stands for the rank of the $t$-th projected generator matrix at the bottom layer.
\end{proposition}
\begin{proof}
	Please refer to Appendix \ref{app_comp}.
\end{proof}

}

\section{Encoding Insights and Ad-Hoc Projection Pruning}\label{sec_insights}
\subsection{Encoding Insights}\label{sec_encoding}
Although the main objective of this paper is to develop low-complexity schemes for decoding RM subcodes, meanwhile, in this subsection, we provide some insights on how the design of the encoder can affect the decoding complexity as well as the performance. 
Throughout the paper, we define the signal-to-noise ratio (SNR) as ${\rm SNR}:=1/(2\sigma^2)$ and the energy-per-bit $E_b$ to the noise ratio as $E_b/N_0:=n/(2k\sigma^2)$, where $\sigma^2$ is the noise variance. Additionally, the number of outer iterations for our recursive algorithms is set to $N_{\rm max}=3$ to ensure the convergence of the algorithms. 
In this section, we mainly present the results for relatively short RM subcodes in order to have the ability to obtain the MAP decoding performance for additional insights and comparison. In Section \ref{sec_nn}, we present the results for relatively larger RM subcodes.

First, in order to further highlight the efficiency of RM subcodes, in Fig. \ref{fig1}, we compare the block error rate (BLER) performance of RM subcodes with the performance of time-sharing (TS) between RM codes under the optimal MAP decoding.
We consider two RM subcodes with parameters $(n,k)=(64,14)$ and $(64,18)$. The generator matrix construction for these codes is based on having the largest ranks for the projected generator matrices (i.e., $\boldsymbol{G}_{\rm max}$) which will be clarified at the end of this subsection. The TS performance is obtained by assuming that the transmitter employs an $\mathcal{RM}(6,2)$ encoder in $\alpha$ fraction of time and an $\mathcal{RM}(6,1)$ encoder in the remaining $(1-\alpha)$ fraction. In this experiment, we set $\alpha=7/15$ and $11/15$ to achieve the same code rates of $14/64$ and $18/64$, respectively,  as the RM subcodes. It is observed that the RM subcodes with the rates $14/64$ and $18/64$ achieve more than $1$ \si{dB} and $0.4$ \si{dB} gains, respectively, compared to the TS counterparts. Also, the performance of the RM subcode with rate $18/64$ is almost $0.2$ \si{dB} better than the performance of the lower rate code with TS. Note that all the simulation results in this paper are obtained from more than $10^5$ trials of random codewords (except $\mathcal{RM}(6,2)$ under the MAP decoding that has $10^4$ trails). 
\begin{figure}[t]
	\centering
	\includegraphics[trim=0.5cm 0.2cm 0 0,width=3.8in]{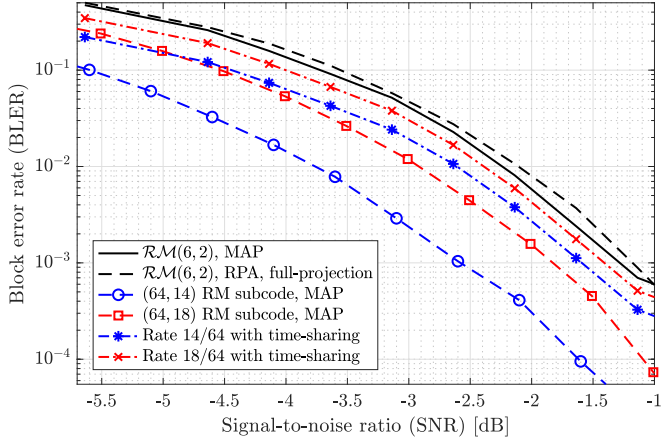}
	\caption{Simulation results for the BLER of various codes under the MAP decoding. The comparison with the time-sharing scheme between $\mathcal{RM}(6,1)$ and $\mathcal{RM}(6,2)$ to achieve the same rates $14/64$ and $18/64$ is also included.}
	\label{fig1}
\end{figure}

As discussed earlier, our decoding algorithms perform the MAP or soft-MAP decoding at the bottom layer. Also, the dimension of the projected codes at the bottom layer (i.e., the rank of the projected generator matrices) can be different. This is in contrast to the RM codes that always result in the same dimension for the projected codes at the bottom layer. Therefore, an immediate approach for encoding RM subcodes to achieve a lower decoding complexity is to  construct the code generator matrix such that the projected codes at the bottom layer have smaller dimensions, and thus the decodings at the bottom layer have lower complexities. In other words, let $L:=\sum_{t=1}^{T}2^{R_{t}}$ represent a rough evaluation of the decoding complexity at the bottom layer, i.e., the decoding complexity at the bottom layer is roughly a constant times $L$. Then, among all $\binom{k_u-k_l}{k-k_l}$ possible selections of the generator matrix $\boldsymbol{G}_{k\times n}$, we can choose the ones that achieve a smaller $L$. This encoding scheme leads to reduction in the decoding complexity of our algorithms but it can also affect the performance. 

\begin{figure}[t]
	\centering
	\includegraphics[trim=0.5cm 0.2cm 0 0,width=3.8in]{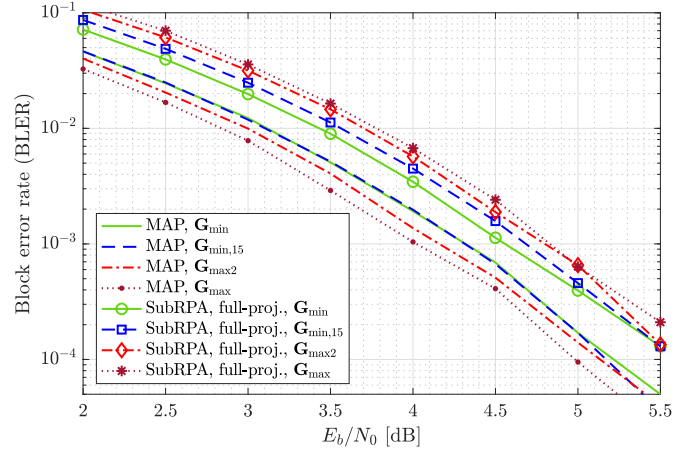}
	\caption{Simulation results for the $(64,14)$ RM subcodes under the MAP and subRPA decoding given four different selections of the generator matrix $\boldsymbol{G}_{k\times n}$.}
	\label{fig2}
\end{figure}
In order to investigate the effect of the aforementioned encoding methodology, in Fig. \ref{fig2}, we consider four different selections of the generator matrix for the $(64,14)$ RM subcode. In particular, $\boldsymbol{G}_{\rm max}$ and $\boldsymbol{G}_{{\rm max}2}$ have the first and second largest values of $L=2568$ and $2532$, respectively, among all possible selections, while $\boldsymbol{G}_{\rm min}$ having the minimum value of $L=1482$. Also, $\boldsymbol{G}_{{\rm min},15}$ has the minimum value of $\sum_{t}2^{R_{t}}=108$ on $15$ projections but a relatively large value of $L=2412$ on all $63$ projections. 
Fig. \ref{fig2} suggests a slightly better performance under the MAP decoder for larger values of $L$. However, surprisingly, our decoding algorithm exhibits a completely opposite behavior, i.e., a better performance is achieved for our subRPA algorithm with smaller values of $L$. This is then a two-fold gain: a better performance for an encoding scheme that results in a lower complexity for our decoding algorithm. We did extensive sets of experiments which all confirm this \textit{empirical} observation. However, still, further investigation is needed to precisely characterize the performance-complexity trade-off as a result of the encoding process.

\subsection{Ad-Hoc Projection Pruning}\label{sec_projprun}
One direction for reducing the complexity of our decoding algorithms is to prune the number of projections at each layer. Particularly, let us assume that, at each layer and node in the decoding tree, the complexity of decoding each branch (that corresponds to a given projection) is the same. This is not precisely true given that the projected codes at the bottom layer may have different dimensions. 
Now, assuming the complexity of the aggregations performed at each layer is the same, 
pruning the number of projections by a factor $\beta\in(0,1)$ is roughly equivalent to reducing the complexity by a factor of $\beta$ at each layer.
In other words, if we have a subcode of $\mathcal{RM}(m,r)$, then there are $r-1$ layers in the decoding tree and hence, the projection pruning exponentially reduces the decoding complexity by a factor of $\beta^{r-1}$. This is essential to make the decoding of higher order RM subcodes practical. One can also opt to choose a constant number of projections per layer (i.e., prune the number of projections at upper layers with smaller $\beta$'s) to avoid high-degree polynomial complexities.

Given that the projected codes at the bottom layer can have different dimensions (in contrast to RM codes), the projection subspaces should be carefully selected to reduce the complexity without having a notable effect on the decoding performance. Our empirical results show that the choice of the sets of projections can significantly affect the decoding performance of RM subcodes.
To see this, in Fig. \ref{fig3}, we consider the generator matrix $\boldsymbol{G}_{{\rm min},15}$ for encoding a $(64,14)$ RM subcode. In addition to full-projection decoding (i.e., $63$ 1-D subspaces), we also evaluate the performance of subRPA and soft-subRPA with $15$ projections picked according to three different projection pruning schemes. 

First, we consider a subset of $15$ subspaces that results in maximum ranks for the projected generator matrices at the bottom layer. In this setting, denoted by ``maxRank'' in Fig. \ref{fig3}, all the $15$ projections result in the same rank of $6$. It is observed that this selection of the projections significantly degrades the performance (almost $1$ \si{dB} gap with full-projection decoding). Our extensive simulation results with other generator matrices and code parameters also confirm the same observation that, although it requires a higher complexity for the MAP or soft-MAP decoding of the projected codes at the bottom layer, the maxRank selection fails to achieve a good performance compared to the other considered projection pruning schemes.

Next, we consider the other extreme of projection selection, i.e., we select $15$ subspaces that result in minimum ranks for the projected codewords. This proposed method for the selection of projections is referred to as  the ``minRank'' scheme in  this paper. In this case, three of the ranks are equal to $2$ and the remaining are equal to $3$. Therefore, the decoder in this case can perform the MAP and soft-MAP decodings at the bottom layer almost $9$ times faster than the maxRank selection (note that $L=108$ and $960$ for the minRank and maxRank selections, respectively). Surprisingly, despite its lower complexity compared to the maxRank selection, the minRank selection is capable of achieving very close to the performance of the full-projection decoding ($\approx 0.1$ \si{dB} gap in the case of both the subRPA and soft-subRPA decoding). Our additional simulation results -- some of which presented in Section \ref{sec_nn} -- \textit{mostly} confirm the same observation and suggest a promising performance for the minRank projection pruning scheme or schemes that result in relatively low $L$'s (if not the minimum $L$).


\begin{figure}[t]
	\centering
	\includegraphics[trim=0.5cm 0.2cm 0 0,width=3.8in]{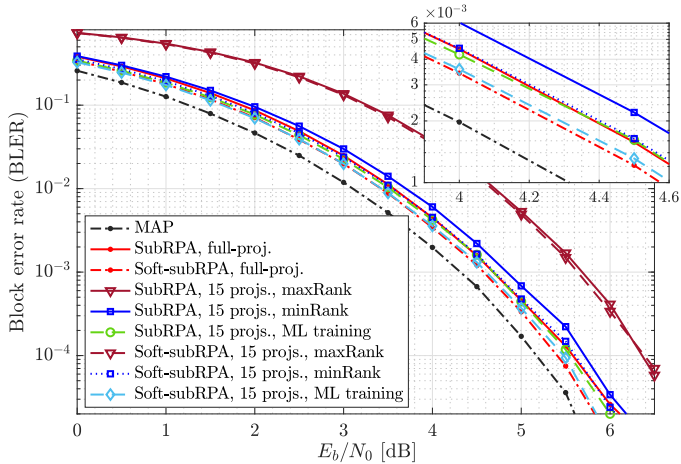}
	\caption{Performance of subRPA and soft-subRPA under full-projection decoding as well as different projection pruning schemes, i.e., picking according to the minimum ranks, maximum ranks, and training a machine learning model. The generator matrix $\boldsymbol{G}_{{\rm min},15}$ is considered for the encoding process of a $(64,14)$ RM subcode.}
	\label{fig3}
\end{figure}

Even though the minRank selection scheme is capable of achieving very close to the performance of full-projection decoding, one cannot guarantee that it is the best selection in terms of minimizing the decoding error rate. In practice, we may want to prune most of projections per layer to allow efficient decoding at higher rates (equivalently, higher order RM subcodes) with a manageable complexity. In such scenarios, we may, inevitably, have a meaningful gap with full-projection decoding, more than what we observed here for minRank selection (i.e., $\approx 0.1$ \si{dB}). Therefore, one needs to ensure that the sets of the selected projections are the ones that minimize the decoding error rate, i.e., the gap to the full-projection decoding. As we will show in Section \ref{sec_nn}, there are scenarios where the performance of the minRank selection significantly diverges from that of the full-projection decoding performance, and it may even perform worse than a random selection of the projections. The failure of the ad-hoc projection pruning schemes in guaranteeing a good performance is the major motivation behind  our ML-aided projection pruning scheme presented in the next section.

 In the next section, we shed light on how the proposed soft-subRPA algorithm enables training an ML model to search for the optimal set of projections. This will then establish the fact that the combination of our soft-subRPA decoding algorithm with our ML-aided projection pruning framework enables efficient decoding (in terms of both decoding error rate and complexity) of RM subcodes. To see the potentials of this scheme, in Fig. \ref{fig3}  the results of our decoding algorithms with $15$ projections picked by training our ML model are also included. It is observed that the trained model also has the tendency to pick projections that result in smaller ranks for the projected generator matrices, i.e., $3$ rank-$2$, $6$ rank-$3$, and $6$ rank-$4$ projections are picked by the ML model (resulting in $L=156$).
Fig. \ref{fig3} demonstrates identical performance to full-projection decoding, for both subRPA and soft-subRPA algorithms, which is the best one can hope for with the pruned-projection decoding. Additionally, it is observed that the soft-subRPA algorithm can  improve upon the performance of the subRPA algorithm by almost $0.1$ \si{dB}.

\begin{figure*}[t]
	\centering
	\includegraphics[trim=0cm 3cm 0 2cm,width=7in]{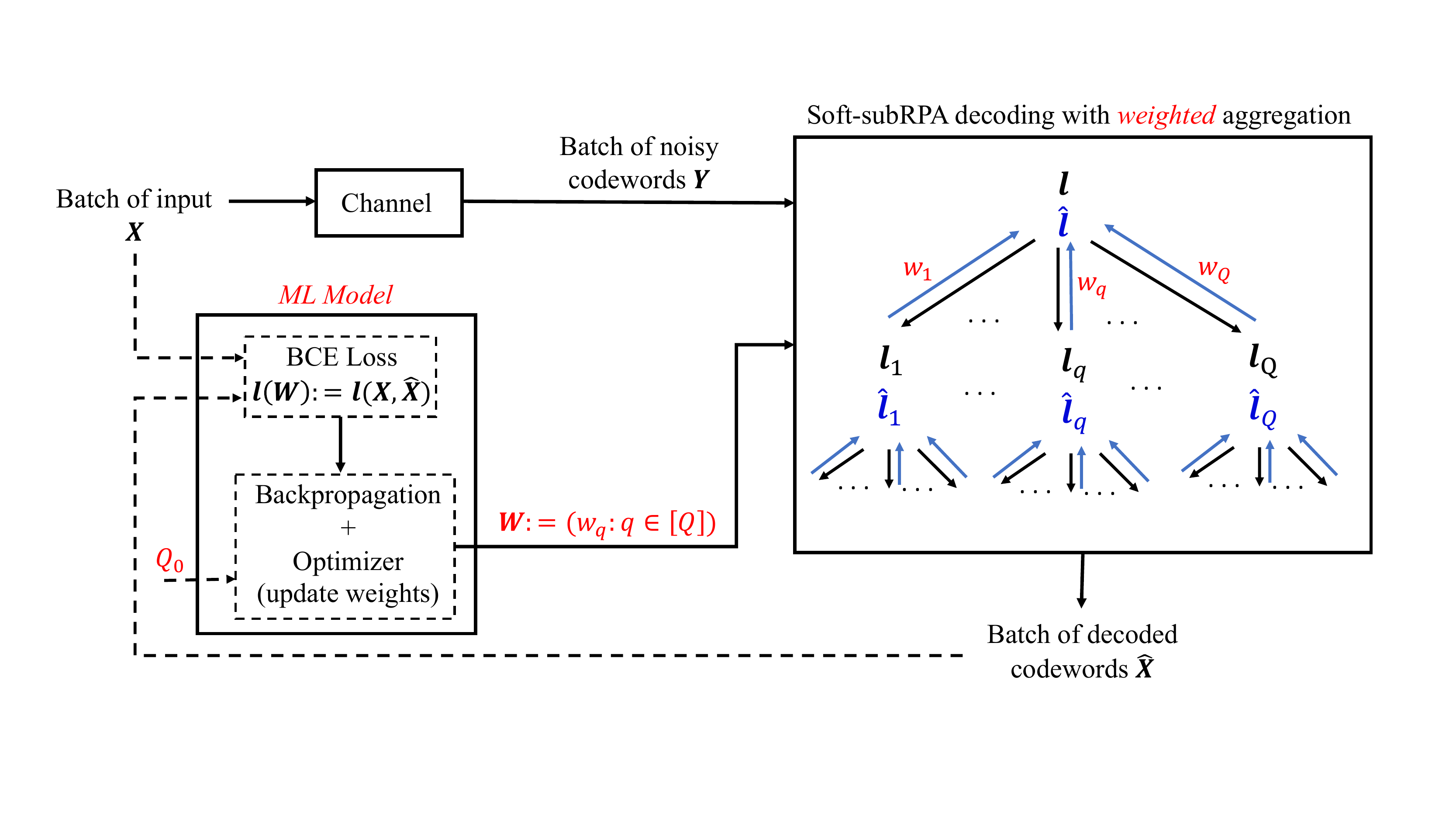}
	\caption{The training procedure of the proposed ML-aided projection pruning scheme for decoding RM subcodes.}
	\label{fig0_2}
\end{figure*}

\section{ML-Aided Projection Pruning}
\label{sec_nn}
As mentioned earlier, the goal is to train an ML model to find the best subset of projections. To do so, as schematically shown in Fig. \ref{fig0_1}, we assign a weight metric $w_q$ to each $q$-th projection such that $w_q\in[0,1]$ and $\sum_{q=1}^Qw_q=1$, where $Q$ is the number of full projections for a given (projected) code in the decoding process. The objective is then to train an ML model  to pick a subset of $Q_0$ projections (i.e., prune the number of projections by a factor $\beta=Q_0/Q$) that minimize the training loss. Building upon the success of stochastic gradient descent methods in training complex models, we want to use gradients for this search. In other words, the ML model updates the weight vector $\boldsymbol{w}:=(w_q, q\in[Q])$ such that picking the $Q_0$ projections corresponding to the largest weights results in the best performance.

There are two major challenges in training the aforementioned ML model. First, the MAP decoding that needs to be performed at the bottom layer (see \eqref{map}) is not differentiable since it involves the ${\rm argmax}(\cdot)$ operation which is not a continuous function. Therefore, one cannot apply the gradient-based training methods to our subRPA algorithm. However, the proposed soft-subRPA algorithm overcomes this issue by replacing the non-differentiable MAP decoder at the bottom layer with the differentiable soft-MAP decoder\footnote{Note that the soft-MAP algorithm involves $\max(\cdot)$ function which, unlike ${\rm argmax}(\cdot)$, is a continuous function. Also, the derivative of the function $\max(0,x)$ is defined everywhere except in $x=0$ which is a rare event to happen. Accordingly, advanced training tools, such as PyTorch library (that is used in this research), easily handle and treat $\max(\cdot)$ as a differentiable function. For example, the rectified linear unit function ${\rm ReLU}(x):=\max(0,x)$ is a widely used activation function in  neural networks.}. 
The second issue is that the combinatorial selection of $Q_0$ largest elements of the vector $\boldsymbol{w}$ is not differentiable. To address this issue, we apply the SOFT (Scalable Optimal transport-based diFferenTiable) top-$k$ operator,
proposed very recently in \cite{xie2020differentiable}, to obtain a smoothed approximation of the top-$k$ operator whose gradients can be efficiently approximated. It is worth mentioning that the SOFT top-$k$ function is a generalization of the soft-max function, which is a soft version of the  $\rm{argmax}$ function. In other words, the SOFT top-$k$ function can be viewed as a soft version of the top-$k$ function.

The training procedure is schematically shown in Fig. \ref{fig0_2}, and is briefly explained next. We use the PyTorch library of Python to first implement our soft-subRPA decoding algorithm in a fully differentiable way for the purpose of the gradient-based training.
We initialize the weight vector as $\boldsymbol{w}_0:=(1/Q,\cdots, 1/Q)$, i.e., equal weights for all the projections. For each training iteration, we randomly generate a batch of $B$ codewords of the RM subcode, and compute their corresponding LLR vectors given a carefully chosen training SNR. Then we input these LLR vectors to the soft-subRPA decoder to obtain the soft decisions at each layer. During the soft-aggregation step, instead of unweighted averaging of \eqref{eq_softaggr}, the weighted averages of the soft decisions at all $Q$ projections are computed as
\begin{align}\label{weighted_avrg}
\tilde{\boldsymbol{l}}(\boldsymbol{z})=\sum_{q=1}^{Q}w_q\tanh\big(\boldsymbol{\hat{l}}_q\left([\boldsymbol{z}+\mathbb{B}_q]\right)/2\big) \boldsymbol{l}(\boldsymbol{z}\oplus\boldsymbol{z}_q).
\end{align}

Ideally, the top-$k$ operator should return nonzero weights only for the top $Q_0$ elements. However, due to the smoothed SOFT top-$k$ operator, all $Q$ elements of $\boldsymbol{w}$ may get nonzero weights though the 
major accumulation of weights will be on the largest $Q_0$ elements once the training is completed.
Therefore, the above weighted average is approximately equal to the weighted average over the largest $Q_0$ weights (i.e., \eqref{weighted_avrg} represents a proper approximation of the aggregation in the case of the pruned-projection decoding).
Note that we apply the same procedure for all (projected) RM subcodes at each node and layer of the recursive decoding algorithm while we define different weight vectors (and possibly different $Q_0$'s) for each sets of projections corresponding to each (projected) codes. We also consider fixed weight vectors for decoding all $B$ codewords at each iteration.

Once the soft decoding of the codewords are obtained, the ML model  updates all weight vectors at each iteration to iteratively minimize the training loss. To do so, we apply the ``Adam'' optimization algorithm \cite{kingma2014adam} to minimize the training loss while using ``BCEWithLogitsLoss'' \cite{nn_loss2} as the loss function, which efficiently combines a sigmoid layer with the binary cross-entropy (BCE) loss. By computing the loss  between the true labels from the generated codewords and the predicted LLRs from the decoder output, the optimizer then moves one step forward by updating the model, i.e., the weight vectors. 

Finally, once the model converges after enough number of iterations, we save the weight vectors to perform optimal projection pruning. Note that in order to reduce the decoding complexity and the overload of training process, we only train the model for a given, properly chosen, training SNR. In other words, once the training is completed, we fix the decoder by picking only the subsets of projections according to the largest values of the weight vectors. We then test the performance of our algorithms given the fixed decoder (i.e., the fixed subsets of projections) for all codewords and across all SNR points. One can apply the same procedure to train the model for each SNR point, or even actively for each LLR vector, to possibly improve upon the performance of our \textit{fixed} projection pruning scheme at the expense of increased model complexity and training overload. 

{ The training SNR, which will be used to generate noisy codewords as training data, is an important hyper-parameter that needs to be carefully chosen to ensure a good performance.
In the context of training models for channel coding, it is conventional to consider a smaller training SNR for the decoder training schedule compared to the encoder training schedule, as the former is often a more challenging task than the latter. 
It is also possible to consider a range of training SNR to further help the single trained model to generalize and perform well across a wide range of SNR during the inference phase (see, e.g., \cite{jamali2023productae} for a thorough empirical investigation on how the training SNR affects the training performance of channel encoders and decoders). In this paper, 
we use a single SNR point (not a range) for training the model to prune the decoding projections. 
We use the result of the full-projection pruning as a benchmark to select the training SNR point (by considering the pruning effect). Specifically, if the full-projection pruning requires $\gamma$ dB to achieve the BLER of $10^{-3}$, we pick the training SNR as $\gamma+\epsilon$ dB, for some positive offset $\epsilon$ that needs to be adjusted according to the pruning factor (i.e., $\epsilon$ is larger if a larger fraction of projections are pruned). Note that this heuristic approach is to pick a starting training SNR, and the final training SNR may need to be adjusted by further hyper-parameter tuning.
}

\begin{figure}[t]
	\centering
	\includegraphics[trim=0.5cm 0.2cm 0 0,width=3.8in]{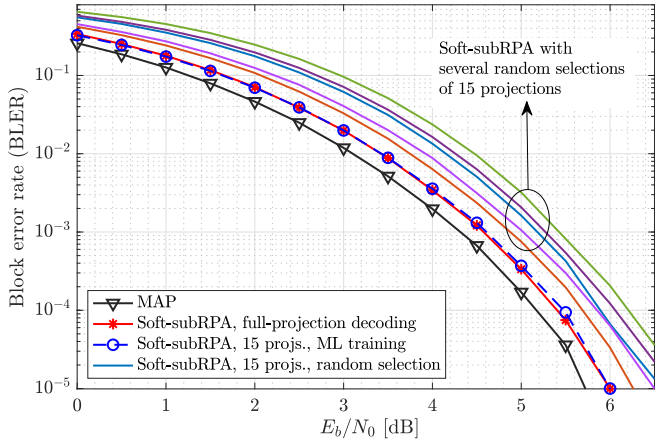}
	\caption{Performance comparison of the MAP decoder with full- and pruned-projection soft-subRPA decoding for a $(64,14)$ RM subcode encoded using the generator matrix $\boldsymbol{G}_{{\rm min},15}$. The performance of the ML-aided projection pruning is also compared to several random selections of projections.}
	\label{fig0}
\end{figure}
Fig. \ref{fig0} demonstrates the potentials of our ML-aided soft decoding algorithm, i.e., soft-subRPA with ML-aided projection pruning, in efficiently decoding RM subcodes. In this experiment, $\boldsymbol{G}_{{\rm min},15}$ 
is used to encode a $(64,14)$ RM subcode.\footnote{We should emphasize that the proposed decoding algorithms and the ML-aided projection pruning scheme are presented in general forms and are not restricted to low rates and lengths.
While decoding a higher-order RM subcode requires a higher complexity, the ML-aided pruning scheme reduces the complexity by a factor of $\beta^{r-1}$ ensuring the best decoding performance given a pruning factor. 
 In our numerical experiments, we focus on subcodes of order-$2$ RM codes that correspond to relatively small code dimensions (i.e., low rates).
  This should not be interpreted as a limitation of our schemes.}
It is observed that our ML-based projection pruning scheme, with only $15$ projections, is able to achieve an almost identical performance to that of the full-projection soft-subRPA decoding with $63$ projections. This is equivalent to reducing the complexity by a factor of more than $4$ without sacrificing the performance. Our low-complexity ML-based pruned-projection decoding has then only about $0.25$ \si{dB} gap with the performance of the MAP decoding. For comparison, the performance of the pruned-projection decoding under several random selections of $15$ projections is also provided. As seen, the choice of the projections can significantly impact the decoding performance of RM subcodes, and 
 randomly selecting the subsets of projections 
cannot guarantee a
 competitive performance.

\begin{figure}[t]
	\centering
	\includegraphics[trim=0.5cm 0.2cm 0 0,width=3.8in]{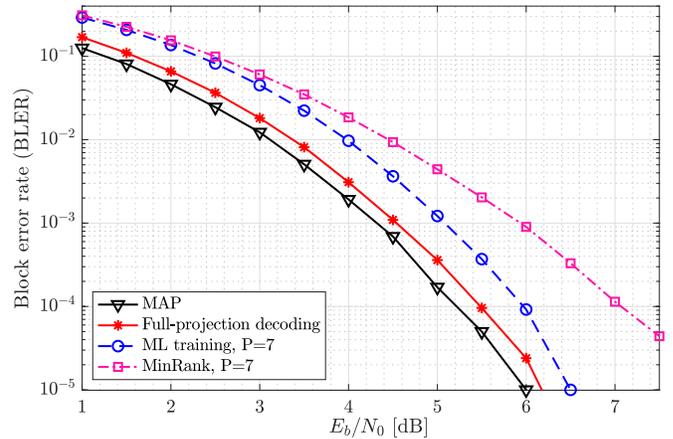}
	\caption{Performance comparison of the MAP decoder with the soft-subRPA decoding for a $(64,14)$ RM subcode encoded using the generator matrix $\boldsymbol{G}_{{\rm min}}$. Full-projection decoding and pruning with $P=7$ projections are considered.}
	\label{fig7}
\end{figure}
Fig. \ref{fig7} presents the performance of a $(64,14)$ RM subcode encoded using the generator matrix $\boldsymbol{G}_{{\rm min}}$. Pruned-projection soft-subRPA decoding with very small number of projections, i.e., $P=7$, is considered. The ML-aided projection-pruned decoding, with $9$ times smaller number of projections, is observed to have less than $0.4$ dB gap with the full-projection decoding. However, the minRank selection significantly degrades the performance, resulting in more than $1$ dB gap with the ML-aided pruning scheme at the BLER of $10^{-4}$. 
To train the ML model in Fig. \ref{fig7}, $Q_0$ was set to $5$ during the training phase but $P=7$ projections corresponding to the largest $7$ weights were selected for the testing. The rationale behind this selection was that nearly $20$\% of the weights were distributed outside the largest $5$ weights (due to the SOFT top-$k$ function), as the sorted  weight vector after training was $\boldsymbol{w}_{\rm sorted}=[0.2012, 0.1781, 0.1519, 0.1444, 0.1279, 0.1277, 0.0689, 0.0000,\\\cdots, 0.0000]$. 
Out of $63$ projected generator matrices of $\boldsymbol{G}_{{\rm min}}$, there are $1$ with rank $1$, $2$ with rank $2$, $28$ with rank $4$, and $32$ with rank $5$. Therefore, the projections picked by the minRank selection scheme result in the set of ranks $\{1,2,2,4,4,4,4\}$. The ML-based selection scheme, however, is observed to pick projections that result in the set $\{2,4,4,4,4,4,4\}$ of ranks, implying that some useful information may be lost if the decoder just picks the projections corresponding to minimum ranks (and thus some higher-rank projections are needed) when 
a significant fraction of projections are pruned{\footnote{We should emphasize that this does not mean that the ML-based selection scheme favors higher rank projections. Indeed, our extensive experiments suggest that the ML-based selection mostly favors smaller-rank projections. Specifically, it either results in the same set of ranks as the minRank selection or only substitutes some very low-rank projections with (slightly) higher-rank projections.}}.

\begin{figure}[t]
	\centering
	\includegraphics[trim=0.5cm 0.2cm 0 0,width=3.8in]{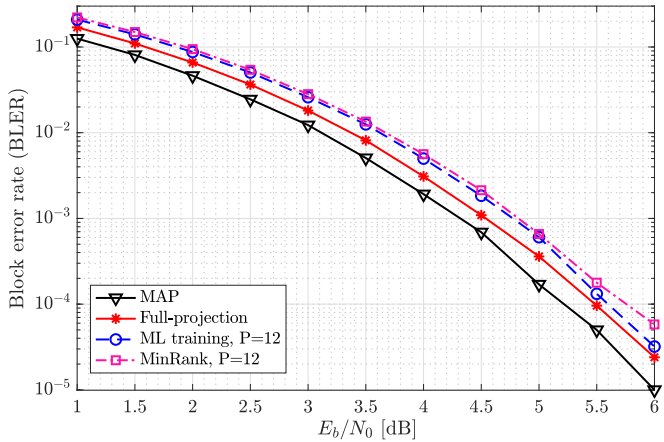}
	\caption{Performance of the full- and pruned-projection ($P=12$ projections) soft-subRPA decoding of a $(64,14)$ RM subcode encoded using $\boldsymbol{G}_{{\rm min}}$.}
	\label{fig8}
\end{figure}
Fig. \ref{fig8} shows the performance of a $(64,14)$ RM subcode, encoded using $\boldsymbol{G}_{{\rm min}}$, under the MAP and soft-subRPA decoding. The ML training was performed under $Q_0=7$ projections. However, since there were $12$ projections with much larger weights,  $P=12$ projections are considered for the testing plots of the pruned-projection decodings in Fig. \ref{fig8}. It is observed that both the minRank and ML-aided pruning schemes achieve very close to the performance of the full-projection decoding, with the ML-aided scheme slightly improving upon the minRank selection at higher SNRs (note that $5\times 10^5$ codewords were used to simulate the performance at each SNR point).
In terms of the rank statistics, it is observed that both selection schemes pick the projections that result in the minimum ranks, i.e., $1$ rank-$1$, $2$ rank-$2$, and $9$ rank-$4$ projections are picked by both schemes. However, the set of the selected projections are still different, as the two schemes only have $6$ projections in common, out of the total $12$ projections. In this case, we can think of the ML model \textit{breaking ties} among the projections that result in the same rank.

{Note that the parameter $Q_0$ is in general a hyper-parameter that needs to be tuned during the training. However, our experiments show that it is not very sensitive, i.e., a model trained for a given $Q_0$ may work well
	for different values of $P$ (i.e., the number of projections during testing/inference). In an ideal case, to use a fixed number $P$ of projections for pruned-projection decoding, one can set the parameter $Q_0=P$ for training. However, this choice may not be the best option. First, 
	due to the SOFT top-$k$ operator,
	we may not observe a sharp drop of trained weights after exactly $Q_0$ largest weights.
	Second, it is possible that some projections are equally good/bad and it is hard for the ML model to perfectly distinguish among them, so the ML model may end up assigning similar weights to such projections. Therefore, to use a fixed $P$, one can train ML models for some larger/smaller values of $Q_0$ than $P$, in addition to $Q_0=P$.
	However, our various training experiments (not presented here) suggest that this hyper-parameter tuning does not much affect the performance of the trained model.
	In the following figure, we use a single model trained for $Q_0=20$ for the selection of both $P=9$ and $P=44$ projections in an RM subcode of parameters $(256,30)$.
}

\begin{figure}[t]
	\centering
	\includegraphics[trim=0.5cm 0.2cm 0 0,width=3.8in]{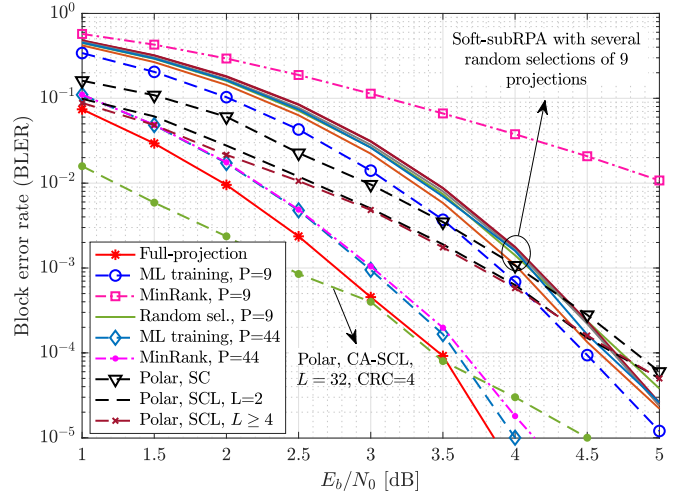}
	\caption{Performance of the full- and pruned-projection soft-subRPA decoding of a $(256,30)$ RM subcode generated through the $\boldsymbol{G}_{{\rm min}}$ encoding. $P=9$ and $44$ projections are considered for the pruned-projection decoding under the minRank, ML-aided, and random pruning schemes. {The plots for the performance of the polar $(256,30)$ code under successive cancellation (SC) decoding, SC-list (SCL) decoding, and cyclic redundancy check (CRC) aided SCL (CA-SCL) are also included.}}
	\label{fig9}
\end{figure}

Fig. \ref{fig9} presents the results for a medium-length RM subcode of parameters $n=256$ and $k=30$ constructed according to the $\boldsymbol{G}_{{\rm min}}$ encoding. To train the ML-aided projection-pruning model, $Q_0$ was set to $20$. However, two different values of $P=9$ and $44$ are used as the number of projections for testing the performance. These selections for $P$ were made by taking into account the profile of the weights after training (picking a $P$ if there is a sharp drop in the value of the next largest weight), and to study two extreme scenarios: 1) a relatively small number of projections such that there is a significant gap to the full-projection decoding; and 2) a relatively large $P$ where the performance of the ML-aided pruned-projection decoding is close to that of the full-projection decoding.

When $P=9$, where the projections are heavily pruned by a factor of more than $28$, the minRank training is observed to significantly diverge from the full-projection decoding performance (e.g., nearly $3$ and $4$ dB gaps at the BLERs of  $10^{-2}$ and $10^{-4}$, respectively). However, training the ML model is shown to enable achieving a significantly better performance. Moreover, the performance of several random selections of the projections are also tested, where, similar to Fig. \ref{fig0}, it is observed that 
the random projection pruning scheme fails to guarantee achieving the best performance for a given value of $P$. On the other hand, when $P=44$ projections are used, both the minRank and ML-aided projection pruning schemes are observed to achieve very close to the performance of the full-projection decoding, with the ML-aided scheme slightly outperforming the minRank scheme at the higher SNRs.


{
Fig. \ref{fig9} also compares the performance of the RM subcode with that of the polar $(256,30)$ code under successive cancellation (SC) decoding and SC-list (SCL) decoding. To construct the polar code, the Tal-Vardy code construction method is used to pick the $k$ bit-channels with the smallest BERs \cite{TV2}. The performance of the cyclic redundancy check (CRC) aided SCL (CA-SCL) decoding of the polar code is also included.
We note, however, that the comparison to the CA-SCL may not be fair as 
one can also do RM-CRC and consider RPA-type decoding algorithms together with Chase list decoding (see, e.g., \cite{ye2020recursive}).
Indeed, 
the comparison of plain codes with plain decoders is more meaningful, and polar with CRC is essentially a concatenated design.} The following are the main conclusions drawn from this figure.
\begin{itemize}
	\item First, the polar code under SC decoding fails to provide a comparable performance to that of the RM subcode, even under $P=7$ projections.
	\item The performance of the polar code under SCL quickly saturates with respect to the list size $L$ such that only a very minimal improvement is observed with increasing $L$, i.e.,  some gains from $L=1$ to $L=2$, very little gain from $L=2$ to $L=4$, and no gain from $L=4$ to larger $L$'s. This is while the RM subcode is able to achieve a much better performance by increasing $P$ from $9$ to $44$.
	\item The RM subcode under $P=44$ is able to achieve a significantly better performance than the polar code under SCL decoding with any list size. Even with $P=9$, the RM subcode beats the polar code under SCL for BLERs smaller than $\approx 7\times 10^{-4}$.
\end{itemize}
It is worth noting that, as seen in Fig. \ref{fig2}, the performance of an RM subcode, for a given $k$ and $n$, highly depends on the selection of the rows, i.e., the encoder design. Therefore, the objective of the paper is not to have a better performance than other classes of codes (which necessitates the best design of the RM subcode encoder) but to deliver the best performance given an RM subcode encoder (that, as shown above, has the potential to beat other classes of codes). We shall emphasize that the low latency of our decoding algorithms is another major advantage compared to polar codes as all decoding branches in the decoding tree (see Fig. \ref{fig0_1}) can be executed in parallel.

\begin{figure}[h]
	\centering
	\includegraphics[trim=0.5cm 0.2cm 0 0,width=3.8in]{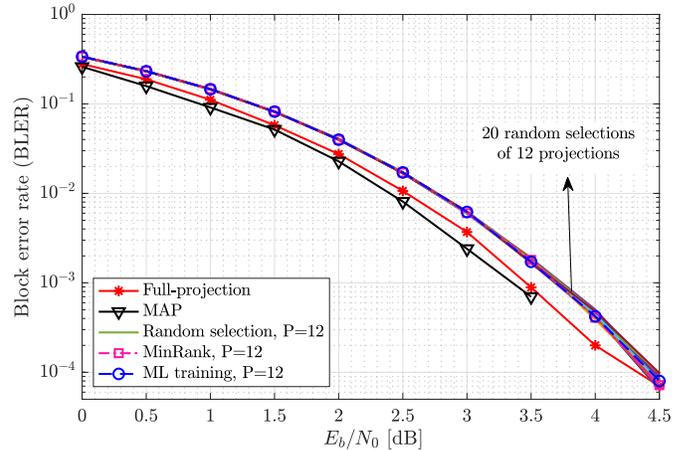}
		\caption{Performance of the full- and pruned-projection soft-subRPA (that reduces to the soft-version of RPA) decoding of an $\mathcal{RM}(6,2)$ code. $P=12$ projections are considered for the pruned-projection decoding under the minRank, ML-aided, and random pruning schemes.}
	\label{fig10}
	\vspace{-0.1in}
\end{figure}

{ In Fig. \ref{fig10}, the (soft-) subRPA decoding algorithm is applied to an $\mathcal{RM}(6,2)$ code (that has $k=22$ and $n=64$). As discussed in Remark 2, in this case, our decoding algorithm reduces to the original RPA decoding of RM codes \cite{ye2020recursive}. By evaluating the performance of many different random selections of $P=12$ projections, it is observed that the performance of a pruned-projection decoding of an RM code, for a given $P$, is (almost) the same regardless of the selection of the projections. This empirical observation then suggests that not much (if any) gain can be expected from ML training for projection selection in RPA decoding of RM codes. As such, our ML-based projection selection as well as the minRank scheme achieved the same performance as random selection of projections. This further suggests that the selection of projections is strongly tied to the rank profile/properties of the so-called projected generator matrices. We believe the theoretical study of this behavior, on both encoding and decoding of RM subcodes, is an interesting direction for future research.}


{ Finally, Fig. \ref{fig11} shows the error probability profile of encoded bits for the sample $(64,14)$ RM subcode with $P=12$ ML-aided projections that corresponds to the setting in Fig. \ref{fig8}. The $E_b/N_0$ is changed from $1$ dB to $4.5$ dB with the step size of $0.5$ dB. For each $E_b/N_0$ point, $10^5$ random codewords are examined and the mismatch of the decoder output with the encoded codeword is evaluated. It is observed that under all evaluated $E_b/N_0$'s, all encoded bits experience a relatively uniform/equal error probability.
}

\begin{figure}[t]
	\centering
	\includegraphics[trim=0.5cm 0.2cm 0 0,width=3.8in]{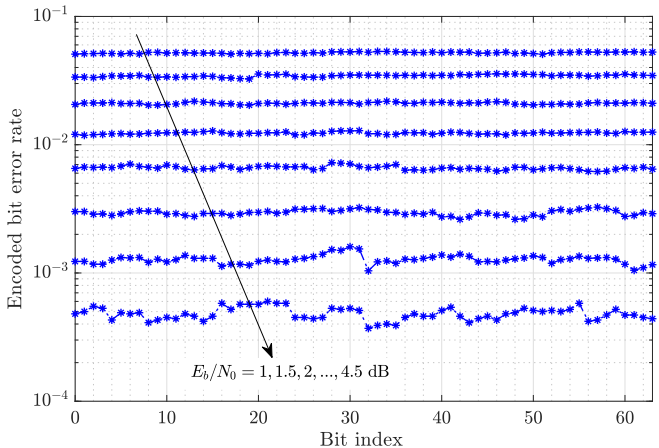}
	\caption{{Encoded bits error rate profile for the $(64,14)$ RM subcode under ML-aided projection selection.}}
	\label{fig11}
	\vspace{-0.1in}
\end{figure}

\section{Conclusions}\label{conc}
In this paper, we designed efficient decoding algorithms for decoding subcodes of RM codes. 
More specifically,
we first proposed a general recursive algorithm, namely the subRPA algorithm, for decoding RM subcodes. Then we derived a soft-decision based version of our algorithm, called the soft-subRPA algorithm, that not only improved upon the performance of the subRPA algorithm but also enabled a differentiable implementation of the decoding algorithm for the purpose of training a machine learning (ML) model. Accordingly, 
we proposed an efficient pruning scheme that finds the best subsets of projections via training an ML model.

Our simulation results on $(64,14)$ and $(256,30)$ RM subcodes demonstrate achieving very close the performance of the full-projection decoding using our ML-aided pruned-projection decoding algorithm with more than $4$ times smaller number of projections.
Our decoding algorithm also inherits the low-latency and parallelized implementation of the RPA algorithm;
when the training is completed, the set of projections are fixed, and all branches in the decoding tree can be executed in parallel.
We also provided some insights on encoding RM subcodes and studied several ad-hoc projection pruning schemes. Our extensive simulations showed that the random selection of projections cannot guarantee a competitive performance to that of the ML-aided pruning scheme, while the proposed minRank pruning scheme being \textit{often} a reasonable structured scheme, especially when the projections are not heavily pruned. On the other hand, when a significant fraction of projections are pruned, the minRank scheme was observed to significantly degrade the performance compared to the ML-aided pruning scheme.

The research in this paper can be extended in several directions such as training ML models to design efficient encoders for RM subcodes, and also leveraging higher dimension subspaces for projections to, possibly, further reduce the decoding complexity.

\appendices
\section{Proof of \Pref{prop_subcode}}\label{app_prop1}
The projection of $\mathcal{RM}(m,r)$ onto a $s$-dimensional subspace, $1\leq s\leq r$ is an $\mathcal{RM}(m-s,r-s)$ code \cite{ye2020recursive}. The code $\mathcal{C}$, that is a subcode of $\mathcal{RM}(m,r)$, is constructed by removing $k_u-k$ rows of the generator matrix of $\mathcal{RM}(m,r)$ that are not in the generator matrix of $\mathcal{RM}(m,r-1)$. We note that the projection of $\mathcal{RM}(m,r-1)$ onto a $s$-dimensional subspace, $1\leq s\leq r-1$, is an $\mathcal{RM}(m-s,r-1-s)$ code. Now, given that each $s$-dimensional projection is equivalent to partitioning $n$ columns of the generator matrix into $n/2^s$ groups of $2^s$ columns and adding them in the binary field (see Remark 1), the generator matrices of the projected codes contain rows of the generator matrix of $\mathcal{RM}(m-s,r-1-s)$ and, possibly, a subset of the rows of the generator matrix of $\mathcal{RM}(m-s,r-s)$ that are not in the generator matrix of $\mathcal{RM}(m-s,r-1-s)$. More precisely, if the selected additional $k-k_l$ rows do not contribute in the rank of the merged matrix according to a given subspace, the projected code onto that subspace is an $\mathcal{RM}(m-s,r-1-s)$ code. On the other hand, if the removed $k_u-k$ rows do not contribute in that rank, the projected code is an $\mathcal{RM}(m-s,r-s)$ code. Otherwise, that projected code is a subcode of $\mathcal{RM}(m-s,r-s)$. \endproof

{ 
	\section{Memory Requirements to Store Projected Matrices in (Soft-) MAP Algorithm}\label{app_memory}
	As discussed in Sections \ref{sec_subRPA} and \ref{sec_Soft-subRPA}, one can pre-compute and store the codebook of each projected code at the bottom layer to facilitate the (soft-) MAP decoding at that layer. In this appendix, we quantify the memory requirement for storing such matrices at the bottom layer, and discuss alternative approaches in applications with limited memory availability.
	
	Recall that for a subcode of $\mathcal{RM}(m,r)$, with $r>1$, the decoding involves $r-1$ layers of 1-D projections, resulting in $T=\prod_{i=1}^{r-1}(\frac{n}{2^{i-1}}-1)=\mathcal{O}(n^{r-1})$ projections for full-projection decoding.
	This number reduces to $T=\mathcal{O}(\beta^{r-1} n^{r-1})$ for a pruned-projection decoding with the pruning factor $\beta<1$.
	After $r-1$ layers of 1-D projections, we arrive at subcodes of $\mathcal{RM}(m-r+1,1)$ whose dimension is $R_t\leq m-r+2$, $\forall t\in[T]$. Therefore, the so-called projected codebooks $\mathcal{C}_p^{(t)}$ will contain $2^{R_{t}}$ (i.e., at most $2^{m-r+2}=n/2^{r-2}$) length-($n/2^{r-1}$) codewords, that can be stored in so-called projected codebook matrices $\boldsymbol{C}_p^{(t)}$ of size at most $(n/2^{r-2})\times(n/2^{r-1})$. Therefore, 
	$\mathcal{O}(\beta^{r-1}n^{r+1}/2^{2r-3})$ bits are required to store all $T$ projected codebooks. For example, for subcodes of $\mathcal{RM}(6,2)$ and $\mathcal{RM}(8,2)$ with $\beta=7/63$ and $9/255$ (that correspond to Figs. \ref{fig7} and \ref{fig9}, respectively), at most 14,563 and 296,068 bits (i.e.,  nearly $1.82$ kB and $37$ kB) respectively, are needed to store all codebooks at the bottom layer. 
	
	Similarly, $\mathcal{O}(k\beta^{r-1}n^{r}/2^{r-1})$ bits are needed to store all $T$ {projected generator matrices} $\boldsymbol{G}_p^{(t)}$ of dimension $k\times 2^{m-r+1}$. Finally, since each matrix $\boldsymbol{U}_p^{(t)}$ is of size $2^{R_t}\times k$,
	$\mathcal{O}(k\beta^{r-1}n^{r}/2^{r-2})$ bits are also needed to store all $T$ matrices $\boldsymbol{U}_p^{(t)}$. 
	Therefore, the memory $M_{\rm tot}$ (in terms of the number of bits) required to store all matrices $\boldsymbol{C}_p^{(t)}$, $\boldsymbol{U}_p^{(t)}$, and $\boldsymbol{G}_p^{(t)}$, $\forall t\in[T]$, at the bottom layer can be characterized as
	\begin{align}\label{memory}
	M_{\rm tot}&=\mathcal{O}(\beta^{r-1}n^{r+1}/2^{2r-3})+\mathcal{O}(k\beta^{r-1}n^{r}/2^{r-1})\nonumber\\
	&{\hspace{0.39cm}}+\mathcal{O}(k\beta^{r-1}n^{r}/2^{r-2})\nonumber\\
	&=\mathcal{O}\left(\beta^{r-1}n^{r}\left[3k+n/2^{r-2}\right]/2^{r-1}\right)
	\nonumber\\
	&=\mathcal{O}\left((n\beta/2)^{r}\left[k+n/2^{r-2}\right]\right).
	\end{align}
	Note that the pruning factor $\beta$ can be essentially $\mathcal{O}(1/n)$ so that the number of projections in each layer, i.e., $\mathcal{O}(\beta n)$,  becomes a constant. Then $(\beta n)^r=\mathcal{O}(1)$ (though with a large constant) and the overall memory requirement will scale linearly with $n$.
	
	Given the above analysis, in applications where this memory requirement may be hard to satisfy, one can directly apply Algorithm \ref{U_Cp_Alg} to compute these matrices during the decoding. We would like to emphasize that the use of the soft-MAP decoding at the bottom layer is motivated by the fact that all projected codewords are subcodes of order-1 RM codes whose dimensions are $R_t\leq m-r+2$. Given that our experiments suggest that projections with smaller $R_t$ are favorable in the decoding process, the above matrices are often significantly smaller than the bounds analyzed here, and the soft-MAP algorithm can be easily afforded. Nevertheless, one may extend the lower-complexity fast Hadamard transform (FHT) decoder of order-1 RM codes 
	to subcodes of order-1 RM codes, and then apply the extended FHT algorithm (instead of MAP) in the subRPA or its soft version (instead of soft-MAP) in the soft-subRPA algorithm or for training the ML model.}

\section{LLRs of the Information Bits}\label{appnd_LLR_inf}
Consider an AWGN channel model as $\boldsymbol{y}=\boldsymbol{s}+\boldsymbol{n}$, where $\boldsymbol{s}=1-2{\boldsymbol{c}}$, $\boldsymbol{c}\in \mathcal{C}$, and $\boldsymbol{n}$ is the AWGN vector with mean zero and variance $\sigma^2$ elements. Then, the LLR of the $i$-th information bit $u_i$ can be obtained using the max-log approximation as
\begin{align}\label{llr_inf}
\boldsymbol{l}_{\rm inf}(i) \approx \operatorname*{max}_{\boldsymbol{c}\in\mathcal{C}_i^0}~\langle \boldsymbol{l}, 1-2{\boldsymbol{c}}\rangle~-~ \operatorname*{max}_{\boldsymbol{c}\in\mathcal{C}_i^1}~\langle \boldsymbol{l}, 1-2{\boldsymbol{c}}\rangle,
\end{align}
where $\boldsymbol{l}:=2\boldsymbol{y}/\sigma^2$ is the LLR vector of the AWGN channel, and $\mathcal{C}_i^0$ and $\mathcal{C}_i^1$ are the subsets of the codewords whose $i$-th information bit $u_i$ is equal to zero or one, respectively.
To see this, observe that
\begin{align}\label{llrinf1}
\boldsymbol{l}_{\rm inf}(i) :=& \ln\left(\frac{\Pr(u_i=0|\boldsymbol{y})}{\Pr(u_i=1|\boldsymbol{y})}\right)\nonumber\\
\stackrel{(a)}{=}&\ln\left(\frac{\sum_{\boldsymbol{s}\in\mathcal{C}_i^0}\exp\left(-||\boldsymbol{y}-\boldsymbol{s}||_2^2/\sigma^2\right)}{\sum_{\boldsymbol{s}\in\mathcal{C}_i^1}\exp\left(-||\boldsymbol{y}-\boldsymbol{s}||_2^2/\sigma^2\right)}\right)\nonumber\\
\stackrel{(b)}{\approx}&\frac{1}{\sigma^2}\operatorname*{min}_{\boldsymbol{c}\in\mathcal{C}_i^1}||\boldsymbol{y}-\boldsymbol{s}||_2^2-\frac{1}{\sigma^2} \operatorname*{min}_{\boldsymbol{c}\in\mathcal{C}_i^0}||\boldsymbol{y}-\boldsymbol{s}||_2^2,
\end{align}
where step $(a)$ is by applying the Bayes' rule, the assumption $\Pr(u_i=0)=\Pr(u_i=1)$, the law of total probability, and the distribution of Gaussian noise. Moreover, step $(b)$ is by the max-log approximation. Finally, given that all $\boldsymbol{s}$'s have the same norm, we obtain \eqref{llr_inf}.

{
\section{Proof of \Pref{prop_comp}}\label{app_comp}
It is well known that the decoding complexity of the full-projection RPA-like decoding of an $\mathcal{RM}(m,r)$ code is $\mathcal{O}(n^r \log n)$ \cite{ye2020recursive}. Similarly, a proof by induction can show that the decoding complexity of our algorithms for a subcode of an $\mathcal{RM}(m,r)$ code, $r>1$,  is $\mathcal{O}(n^{r-1} \mathcal{C}(m-r+1,1))$, where $\mathcal{C}(m',1)$ stands for the complexity of decoding a subcode of an $\mathcal{RM}(m',1)$ code. We note that (proof by induction) the above complexity reduces to $\mathcal{O}(({\beta n})^{r-1} \mathcal{C}(m-r+1,1))$ for pruned-projection decoding with a pruning factor $\beta<1$. 

We first note that, assuming (soft-) MAP at the bottom layer, $\mathcal{C}(m-r+1,1)$ can be characterized as $\mathcal{O}(n_12^{k_1})$, where $n_1=2^{m-r+1}$ is the code length and $k_1=m-r+2$ is the code dimension in the bottom layer. Therefore, $\mathcal{C}(m-r+1,1)=\mathcal{O}(2^{m-r+1}2^{m-r+2})=\mathcal{O}(n^2/2^{2r-3})$. This complete the proof of the first part, i.e., the $\mathcal{O}(n^{r+1})$ complexity for full-projection decoding.

Next, as discussed in Appendix \ref{app_memory}, the pruning factor $\beta$ can be essentially $\mathcal{O}(1/n)$ so that the number of projections in each layer, i.e., $\mathcal{O}(\beta n)$,  becomes a constant. Then, $(\beta n)^{r-1}=\mathcal{O}(1)$ (though with a large constant) and the overall complexity reduces to $\mathcal{O}(\mathcal{C}(m-r+1,1))$. This then complete the proof of the second part, i.e., $\mathcal{O}(n^{2})$ complexity for pruned-projection decoding with pruning factor $\beta=\mathcal{O}(1/n)$.

Finally, as empirically observed in Section \ref{sec_nn}, in most cases the selected projections by our ML training scheme have very small (nearly the smallest) ranks $R_t$ for the projected generator matrices. Therefore, the number of codewords $2^{k_1}=2^{R_t}$ may be upper bounded by a constant. This then reduces the complexity to $\mathcal{O}(n)$ if $2^{R_t}=\mathcal{O}(1)$, $\forall t\in[T]$, in addition to $\beta=\mathcal{O}(1/n)$.

We would like to emphasize that the complexity analysis above my require some large constants (modeled by $\mathcal{O}(1)$). Therefore, even if the complexity can linearly scale with $n$, the involved constants may be large. However, a major advantage of our decoding algorithms is the reduction in the latency (e.g., compared to polar codes) as all the branches involved in the decoding tree (see, e.g., Fig. \ref{fig0_1}) can be executed in parallel. We refer the readers to \cite{ye2020recursive} for additional discussions on latency aspects of RPA-like decoding of RM codes.

}

\end{document}